\documentclass[11pt]{article}
\usepackage{color}
\usepackage[letterpaper, margin=1in]{geometry}
\definecolor{gray}{rgb}{0.3, 0.3, 0.8}
\definecolor{light-gray}{gray}{0.95}
\definecolor{darkgreen}{rgb}{0.0,0.6,0.0}
\usepackage[colorlinks=true,urlcolor=red, linkcolor=darkgreen,citecolor=darkgreen]{hyperref}
\usepackage{graphicx}
\usepackage{subcaption}
\usepackage{amsthm,amsmath,amsfonts}
\usepackage{latexsym,amssymb}
\usepackage{mathtools}
\usepackage[boxruled]{algorithm2e}
\usepackage{algorithmic}
\usepackage{enumerate}
\usepackage{thmtools,thm-restate}
\usepackage[inline]{enumitem}
\usepackage{cleveref}
\usepackage{lmodern}
\usepackage{bbm}
\AtBeginDocument{\let\phi\varphi}

\newcommand{\Ex}[2]{{\mathbb{E}}_{#1}\left[#2\right]}
\newcommand{\BP}{\mathbb{P}}
\newcommand{\BE}{\mathbb{E}}
\newcommand{\BR}{\mathbb{R}}
\newcommand{\bias}{\mathsf{bias}}
\newcommand{\BH}{\mathbf H}
\newcommand{\BHc}[1]{\mathbf{H}_{#1}}
\newcommand{\eps}{\varepsilon}
\newcommand{\comp}[1]{\overline{#1}}

\newcommand{\CA}{\mathcal A}
\newcommand{\CE}{\mathcal E}
\newcommand{\CZ}{\mathcal Z}
\newcommand{\CD}{\mathcal D}

\newcommand{\CB}{\mathcal B}
\newcommand{\CG}{\mathcal G}

\newcommand{\CS}{\mathcal S}
\newcommand{\CF}{\mathcal F}
\newcommand{\CJ}{\mathcal J}

\newcommand{\CT}{\mathcal T}

\newcommand{\FF}{\mathfrak F}
\newcommand{\CX}{\mathcal X}
\newcommand{\CY}{\mathcal Y}
\newcommand{\supp}{\mathsf{supp}}

\newcommand{\cp}{p}

\newcommand{\sh}{\mathsf h}
\newcommand{\bits}{\{0,1\}}
\newcommand{\Bias}{\mathsf{Bias}}
\newcommand{\str}{\{0,1\}}

\newcommand{\nrank}{\mathsf{rk}_+}

\newcommand{\Eps}{\mathcal E}

\newcommand{\sg}{\mathsf g}
\newcommand{\diag}{\mathsf{diag}}
\newcommand{\ind}{\mathbbm{1}}
\newcommand{\arr}{ - }

\providecommand{\Inf}[2]{\mathbf{I}\left(#1:#2\right)} 
\providecommand{\Infc}[4][]{\mathbf{I}_{#1}\left(#2:#3|#4\right)}

\title{Lower Bounds for Approximating the Matching Polytope}
\author{
Makrand Sinha\thanks{Supported by the National Science Foundation under agreements CCF-1149637, CCF-1420268 and CCF-1524251.}\\
\\\small{Paul G. Allen School of Computer Science \& Engineering,}\\
\small{University of Washington, Seattle}\\
\\{\small \tt makrand@cs.washington.edu}
}
\date{}
\begin{document}

\newtheorem{thm}{Theorem}[section]
\newtheorem{open}{Open Problem}[section]
\newtheorem{theorem}{Theorem}[section]
\newtheorem{claim}[thm]{Claim}
\newtheorem{subclaim}[thm]{Sub-claim}
\newtheorem{corollary}[thm]{Corollary}
\newtheorem{lemma}[thm]{Lemma}
\newtheorem{conj}[thm]{Conjecture}
\newtheorem{prop}[thm]{Proposition}
\newtheorem{proposition}[thm]{Proposition}
\newtheorem{defn}[thm]{Definition}
\numberwithin{equation}{section}
\numberwithin{figure}{section}

\captionsetup[subfigure]{subrefformat=simple,labelformat=simple,listformat=subsimple}
\renewcommand\thesubfigure{(\alph{subfigure})}

\maketitle

\begin{abstract}
	We prove that any extended formulation that approximates the matching polytope on $n$-vertex graphs up to a factor of $(1+\eps)$ for any $\frac2n \le \eps \le 1$ must have at least $\binom{n}{{\alpha}/{\eps}}$ defining inequalities where $0<\alpha<1$ is an absolute constant. This is tight as exhibited by the $(1+\eps)$ approximating linear program obtained by dropping the odd set constraints of size larger than $({1+\eps})/{\eps}$ from the description of the matching polytope. Previously, a tight lower bound of $2^{\Omega(n)}$ was only known for $\eps = O\left(\frac{1}{n}\right)$ \cite{R14, BP15} whereas for $\frac2n \le \eps \le 1$, the best lower bound was $2^{\Omega\left({1}/{\eps}\right)}$ \cite{R14}. The key new ingredient in our proof is a close connection to the non-negative rank of a lopsided version of the unique disjointness matrix.  
	
\end{abstract}

\section{Introduction}

In recent years, there has been a significant development in understanding whether well-known combinatorial optimization problems or polytopes can be expressed as a linear program with a small number of inequalities. The object of interest in this context has been the notion of \emph{extended formulation} of polytopes. A polytope $Q \subseteq \BR^{d'}$ is called an extended formulation of a given polytope $P \subseteq \BR^d$, if there is a linear map $\pi:\BR^{d'} \to \BR^d$ such that  $\pi(Q) = P$. The \emph{extension complexity} of $P$, $\mathsf{xc}(P)$ is defined to be the minimal number of facets in any extended formulation for $P$. If a polytope $P$ with an exponential number of facets has an extended formulation $Q$ of polynomial size, then one can solve an optimization problem over $P$ by writing a small linear program over $Q$. There are known examples of polytopes, such as the spanning tree polytope and the permutahedron (see Section 6 in \cite{CCGZ10}), which have an exponential number of facets but polynomial size extended formulations.

Yannakakis \cite{Y91} initiated the study of extended formulations motivated by refuting a purported $\mathsf{P} = \mathsf{NP}$ proof which encoded the Traveling Salesman problem as a polynomial size linear program. He proved that the extension complexity of polytopes is equal to the non-negative rank of the \emph{slack matrix} associated with the polytope. Using this connection he was able to show that any \emph{symmetric} extended formulation (one whose formulation is invariant under permutation of variables) that projects to the Traveling Salesman and matching polytopes in $\BR^{\binom{[n]}{2}}$ must have size $2^{\Omega(n)}$.

Here, we will be concerned with the matching polytope. The matching polytope $P_{MAT}(n) \subseteq \BR^{\binom{[n]}{2}}$ is defined as the convex hull of all matchings in the complete graph on the vertex set $[n]$ where $[n] := \{1,\dots,n\}$. We will just write $P_{MAT}$ when $n$ is clear from the context. Denoting the indicator vector for a matching $M$ by $\ind_M$, the facets of $P_{MAT}$ are completely described by the following degree and non-negativity constraints, and the \emph{odd set} inequalities, as shown in \cite{E65}:
	\begin{align*}
		\ P_{MAT}(n) &= \mathsf{conv}\left\{\ind_M \in \BR^{\binom{[n]}{2}} ~\bigg|~ M \subseteq \binom{[n]}{2} \text{ is a matching}\right\} \\
		\     &= \left\{x \in \BR^{\binom{[n]}{2}} ~\Bigg|~ \sum_{\substack{e \text{ incident}\\ \text{ on } i}} x_e \le1\; \forall i \in [n]; \sum_{\substack{e \text{ inside }\\ U}} x_e \le \frac{|U|-1}{2}\; \forall U \subseteq [n], |U| \text{ odd} ; x \ge 0\right\}.
	\end{align*}
	Note that the description of the matching polytope has $2^{\Omega(n)}$ facet defining inequalities. However, one can still optimize any linear function in polynomial time over the matching polytope using the algorithm of Edmonds \cite{E65}. The question of whether one could optimize over the matching polytope with a small extended formulation remained open until recently, when in a breakthrough result, Rothvo{\ss} \cite{R14} proved that the extension complexity of $P_{MAT}$ is indeed $2^{\Omega(n)}$.

	The central question that we want to answer here concerns whether the matching polytope can be approximated by a small extended formulation. Formally, we say that a polytope $P$ is monotone if for any $x \in P$ and $y$ satisfying $0 \le y \le x$, we have that $y \in P$. We call a polytope $K$ to be a $(1+\eps)$ approximation of a monotone polytope $P$ if $P \subseteq K \subseteq (1+\eps)P$. Notice that if $P$ is monotone, then the above is equivalent to requiring that for any non-negative vector $w$,
	\[ \max\{w^Tx ~|~x \in P\} \le \max\{w^Tx~|~x \in K\} \le (1+\eps)\max\{w^Tx~|~x\in P\}.\]

The matching polytope $P_{MAT}$ is indeed monotone and it is well-known\footnote{See Appendix \ref{sec:upperboundmat} for a proof.} that for any $\frac2n \le \eps \le 1$, the following polytope $Q_\eps$ with at most $\binom{n}{1+{1}/{\eps}} + O(n^2)$ facets approximates the matching polytope up to a factor of $1+\eps$ (note that for $\eps \le \frac1{n-1}$, $Q_{\eps}$ is the matching polytope itself): 
	\begin{align*}
		 \ Q_\eps(n) = \Bigg\{x \in \BR^{\binom{[n]}{2}} ~\bigg|~ & \sum_{\substack{e \text{ incident}\\ \text{ on } i}} x_e \le1\; \forall i \in [n];\;x \ge 0;\\
		 \                                       &\sum_{\substack{e \text{ inside }\\ U}} x_e \le \frac{|U|-1}{2}\; \forall U \subseteq [n], |U| \text{ odd},\;|U| \le \frac{1+\eps}{\eps}\Bigg\}.
	\end{align*}
	This gives us a \emph{polynomial type approximation scheme} ({PTAS}) style extended formulation (size at most $n^c$ for a constant $c = c(\eps)$ depending on $\eps$) that approximates the matching polytope. Building on the ideas of Rothvo{\ss} \cite{R14}, Braun and Pokutta \cite{BP15} showed that any extended formulation that approximates the matching polytope up to a factor of $1+O\left(\frac{1}{n}\right)$ has size $2^{\Omega(n)}$ ruling out a \emph{fully polynomial type approximation scheme} ({FPTAS}) type extended formulation (size polynomial in both $n$ and $\frac1{\eps}$) for matching. Rothvo{\ss} \cite{R14} observed that this already implies that for any $\frac2n\le\eps\le 1$, any $(1+\eps)$ approximating extended formulation must have $2^{\Omega\left({1}/{\eps}\right)}$ size.
	
	\begin{thm}[\cite{R14, BP15}]
		For any $\frac2n \le \eps \le 1$ and any polytope $P_{MAT} \subseteq K \subseteq (1+\eps)P_{MAT}$,
		\[\mathsf{xc}(K) \ge 2^{\Omega\left({1}/{\eps}\right)}.\] 
  \end{thm}

	Note that the above implies a tight lower bound of $2^{\Omega(n)}$ for $\eps \le \frac2n$ but leaves a gap between the upper and lower bounds. Moreover, the gap gets larger as $\eps$ increases and in particular when $\eps = \Omega(1)$ we do not even have non-trivial lower bounds. 

	The above theorem is proven using the connection between extension complexity and non-negative rank of slack matrices that was established in the work of Yannakakis \cite{Y91} and subsequently extended by Braun, Fiorini, Pokutta and Steurer \cite{BFPS15} to handle approximations of polytopes. The lower bound on extension complexity above, then follows from a lower bound on the non-negative rank of the slack matrix associated with the matching polytope. 

	A matrix that is closely related to the matching slack matrix is the unique disjointness matrix. Let $\CY$ be the collection of all subsets of $[n]$. For a parameter $\rho \in [0,1]$, we say that a non-negative matrix $A \in \BR^{\CY \times \CY}$ is a unique disjointness matrix if 
\begin{equation}
	A_{xy}  = \begin{cases} \;\;\;\;1 \text {, if } |x \cap y| = 0 \text{ and }\\
	 			   		 \le 1-\rho \text{, if } |x\cap y|=1.
			   \end{cases}
	\label{eqn:udisj}
\end{equation}
	
Many recent lower bounds in extended formulations (see Section \ref{sec:history}) follow from a lower bound on the non-negative rank of such matrices. Starting with the breakthrough work of Fiorini, Massar, Pokutta, Tiwary and de Wolf \cite{FMPTW15}, a sequence of works \cite{BFPS15, BM13, BP16} proved that the non-negative rank of any unique disjointness matrix is $2^{\Omega(\rho n)}$. 

For this work, the matrix that will be of relevance is the lopsided version of the unique disjointness matrix where the rows are indexed by $k$-subsets of $[n]$ where $k \le \frac{n}2$, while the columns are indexed by all subsets of $[n]$. This is useful for us since it turns out that to prove an extension complexity lower bound for any $(1+\eps)$-approximation for the matching polytope it suffices to prove a lower bound on similar lopsided slack matrices. The rows of any such slack matrix (an example is exhibited by the slack matrix for the polytope $Q_{\eps}$) are indexed by odd sets of size at most $O(\frac{1}{\eps})$ and the columns are indexed by all possible matchings of which there are exponentially many in $n$.  

From the known lower bounds for the unique disjointness matrix, one could infer that the non-negative rank of the lopsided unique disjointness matrix must be $2^{\Omega(\rho k)}$. One could however, hope to improve this bound to $\approx\binom{n}{k}$ (which is much larger than $2^{\Omega(k)}$ when $k$ is small) at least for the case $\rho=1$ by making use of the lopsided structure. Such lopsided structure has previously been exploited in the setting of communication complexity to prove analogous lower bounds for lopsided versions of disjointness \cite{MNSW98, AIP06, P11, RR15}.  

\subsection{Our Results}

In this paper, we show that the simple upper bound exhibited by the polytope $Q_{\eps}$ defined above is tight in the sense that any extended formulation of a polytope that $(1+\eps)$ approximates $P_{MAT}$ must (roughly) have as many facets as $Q_{\eps}$.

\begin{restatable}{thm}{mat}
	\label{thm:mat}
	For any $\frac2n \le \eps \le 1$ and any polytope $P_{MAT} \subseteq K \subseteq (1+\eps)P_{MAT}$, it holds that
	\[\mathsf{xc}(K) \ge \binom{n}{{\alpha}/{\eps}},\]
  	where $0<\alpha<1$ is an absolute constant. 
\end{restatable}

Note that the above theorem also covers the case for which tight bounds were previously known: when $\eps \le \frac2n$, we get an asymptotically tight lower bound of $\binom{n}{\alpha n/2} = 2^{\Omega(n)}$.

We also prove a tight lower bound on the non-negative rank of the lopsided unique disjointness matrix. Let $k \le \frac{n}2$ and $\CX = \binom{[n]}{k}$ be the collection of $k$-subsets of $[n]$ and let $\CY$ be the collection of all subsets of $[n]$. For a parameter $\rho \in [0,1]$, let $A \in \BR^{\CX \times \CY}$ be any non-negative matrix satisfying \eqref{eqn:udisj}.

Then denoting by $\nrank(A)$ the nonnegative rank of $A$, we prove that: 
\begin{restatable}{thm}{udisj}
   \label{thm:udisj}
   For any $\frac{3\cdot1000^8}{\log(n/k)} \le \rho^8 \le 1$, it holds that $\displaystyle \nrank(A) \ge \binom{n}{\alpha\rho^8 k}$ with $0<\alpha<1$ an absolute constant. 
\end{restatable}

From known results (see Section \ref{sec:acorr}), proving a lower bound on the non-negative rank of any lopsided unique disjointness matrix also gives a lower bound on the extension complexity of a lopsided version of the correlation polytope. Formally, we have the following result for the polytope $P_{LCORR} \in \BR^{n\times n}$ defined as 
\[P_{LCORR} = \mathsf{conv}\left\{bb^T~\bigg|~b\in\bits^n, \sum_{i=1}^n b_i \le k\right\}.\]

\begin{corollary}
	\label{cor:acorr}
	Let $1\le\sigma^8\le\frac{1}{3\cdot1000^8}\log\left(\frac{n}{k}\right)$. For any polytope $P_{LCORR}(n) \subseteq K \subseteq \sigma P_{LCORR}(n)$, it holds that
	$\displaystyle \mathsf{xc}(K) \ge \binom{n}{\alpha k/\sigma^8}$ with $0<\alpha<1$ an absolute constant.
\end{corollary}

\subsection{Other Related Work} 
\label{sec:history}

Here we briefly survey some of the recent related lower bounds on the size of linear programs (LPs) and semidefinite programs (SDPs) for other polytopes and optimization problems. The progress in extension complexity lower bounds was jump-started by the breakthrough work of Fiorini, Massar, Pokutta, Tiwary and de Wolf \cite{FMPTW15} who related the extension complexity of the TSP and Correlation polytopes with the non-negative rank of unique disjointness matrix where the entries are $0$ if the sets intersect and proved exponential lower bounds for the latter. This implied an exponential lower bound on the extension complexity of the TSP polytope extending the results of Yannakakis \cite{Y91} which only applied to symmetric extended formulations. Via reductions established in \cite{AT13,PV13}, exponential lower bounds for extension complexity also follow for several other $\mathsf{NP}$-hard polytopes such as the $\mathsf{3SAT}$ and Knapsack polytopes.

Braun, Fiorini, Pokutta and Steurer \cite{BFPS15} defined a notion of approximation for combinatorial optimization problems in terms of extended formulations. They proved that even approximating the value of any linear objective function over the Correlation polytope requires extended formulations of exponential size by proving non-negative rank lower bounds for unique disjointness when the intersecting entries are $1-\rho$. Braverman and Moitra \cite{BM13} (see also \cite{BP16}) used information theoretic techniques to prove a tight lower bound of $2^{\Omega(\rho n)}$ on the non-negative rank of unique disjointness. An implication of this result was that any LP of size $2^{O(n^\eps)}$ for the convex hull of all cliques in all $n$-node graphs cannot achieve better than an $n^{1-\eps}$ integrality gap. This matches the algorithmic hardness of approximation result of H\aa stad \cite{H96}. Underlying all the works above is a lower bound on the non-negative rank of the unique disjointness matrix. These works, however, did not consider the lopsided version of the unique disjointness matrix. 

In a parallel line of work, Chan, Lee, Raghavendra and Steurer \cite{CLRS13} proved lower bounds on the size of linear programs for constraint satisfaction problems by relating it to known Sherali-Adams integrality gaps using techniques from Fourier analysis and convex optimization. For instance, they proved that any polynomial sized LP for Max-Cut on $n$-node graphs cannot beat the trivial approximation factor of $1/2$. Recently, Kothari, Meka and Raghavendra \cite{KMR17} improved their lower bounds to $2^{n^{c}}$ for some constant $0<c<1$. Lee, Raghavendra and Steurer \cite{LRS15} proved that similar situation arises for SDP relaxations: when it comes to maximum constraint satisfaction problems, SDPs of size polynomial size cannot perform much better than sums of squares relaxations of constant degree. For instance, this implies that for $\mathsf{MAX}$-$3$-$\mathsf{SAT}$, no SDP of polynomial size can beat the trivial approximation factor of $7/8$.

When it comes to approximations, the problems discussed above have exponential size LP lower bounds. The matching problem is very different in the sense that for any $\eps \ge \frac2n$, it has an LP of size $\approx \binom{n}{1/\eps}$ which achieves a $(1+\eps)$ approximation. We show that for the matching problem this is best possible. A similar situation arises for the Max-Knapsack problem for which Bienstock \cite{B08} gave an LP of size $\approx \binom{n}{1/\eps}$ which achieves a $(1-\eps)$ approximation for any $\eps \ge \frac2n$. An exponential lower bound for exact extended formulations for Max-Knapsack follows by a reduction from the unique disjointness matrix \cite{AT13, PV13} (also see \cite{GJW16} which proves a better lower bound by using a different construction). It is unclear how to extend the aforementioned reduction from unique disjointness to prove a strong lower bound for any $(1-\eps)$ approximation for Max-Knapsack. A lower bound even in the case of $(1-\frac1n)$ approximation remains an interesting open problem.

\subsection{Organization}

In Section \ref{sec:techniques} we give the high-level intuition for our results before we delve into technical details. Section \ref{sec:slack} discusses the slack matrix for the Matching and Lopsided Correlation polytopes. Section \ref{sec:preliminaries} contains the basic notation and preliminaries. In Section \ref{sec:lowerbounds} we prove lower bounds on the non-negative rank of the lopsided unique disjointness matrix and the matching slack matrix. Section \ref{sec:mainlemmaproof} proves a key technical lemma.

\section{Overview of Techniques}
\label{sec:techniques}

To give an outline of the proofs, we will assume basic familiarity with the definition of entropy of discrete random variables. Random variables will be denoted by capital letters and the values attained by them will be denoted by smaller letters. We will write $p(x)$ to denote the distribution of $X$ as well as the probability of the event $X=x$ in the probability space $p$, where the meaning will be clear from context. See section \ref{sec:preliminaries} for notational conventions and preliminaries.

Our proof is based on the \emph{common information} approach used by Braun and Pokutta \cite{BP15,BP16}. In this approach, we view a non-negative matrix $L$ indexed by $x$ and $y$ as a probability distribution $p(x,y)$ by normalizing by the total weight of the matrix. Let $X$ and $Y$ be sampled from $p(x,y)$. If $\nrank(L)=r$, then the corresponding distribution $p(x,y)$ is a convex combination of $r$ product distributions given by the non-negative rank one factors. Viewed this way, a non-negative factorization of $L$ gives us a random variable $R$ with support of size $\log \nrank(L)$ that then breaks the dependency between $X$ and $Y$. In other words, $X$ and $Y$ are independent if we know the value of $R$ (see Section \ref{sec:preliminaries}) or formally, $X, R, Y$ form a Markov chain (denoted by $X \arr R \arr Y$).

 \subsection{Lopsided Unique Disjointness}

To give intuition behind the proof of Theorem \ref{thm:udisj}, let us take a look at what the corresponding distribution looks like when the matrix we start with is the lopsided unique disjointness matrix $A$ given by \eqref{eqn:udisj}. Let $X'$ and $Y'$ be the random variables sampled from the distribution obtained by normalizing $A$. Then $X'$ is a random subset of $[n]$ of size at most $k$ and $Y'$ is an arbitrary random subset of $[n]$. It turns out that by using a direct-sum argument (dividing the universe $[n]$ into blocks of size $[\frac{n}{k}]$) and appropriate conditioning we can reduce our question to a problem about random subsets $X$ and $Y$ of the universe $[\frac{n}{k}]$. 
 
In the reduced problem, the set $X \subseteq [\frac{n}{k}]$ is a random subset of size one and $Y \subseteq [\frac{n}{k}]$ is an arbitrary random set while the probability that $X \cap Y = \varnothing$ is at least $\frac12 + \Omega(\rho)$. We will argue that if the non-negative rank was in fact small, then the probability that $X \cap Y = \varnothing$ can not be much larger than $\frac12$ deriving a contradiction. 

In the reduced problem, the entropies of $X$ and $Y$ are given by $\BH(X|X\cap Y=\varnothing) = \log\left(\frac{n}k\right) - O(1)$ and $\BH(Y|X\cap Y=\varnothing) = \frac{n}k - O(1)$. If $\log \nrank(A) \ll \gamma^8 k \log\left(\frac{n}k\right)$ for a small constant $\gamma$, then it turns out that we get a random variable $R$ such that $X - R -Y$ and even after conditioning on $R$ the entropies remain large:
$\BH(X|R, X\cap Y=\varnothing) \ge \log\left(\frac{n}k\right)  -\gamma^8 \log\left(\frac{n}k\right) - O(1)$ and $\BH(Y|R, X\cap Y=\varnothing) \ge \frac{n}{k} - \gamma^8 \log\left(\frac{n}k\right) - O(1)$. Note that the conditioning on $X\cap Y=\varnothing$ is needed to carry out the direct-sum argument and is quite essential.

To prove a non-negative rank lower bound for lopsided unique disjointness, we exploit the lopsided structure to prove the following key technical lemma which intuitively says that for most values of $R$, the probability that the event $X \cap Y=\varnothing$ happens conditioned on $R$ is smaller than $\frac12 + \gamma$. In the following lemma, we view the set $X$ as an element of $[\frac{n}{k}]$ while we view $Y \in \bits^{\frac{n}k}$ as an indicator vector for a subset of $[\frac{n}{k}]$, so $X\cap Y=\varnothing$ is equivalent to the event that $Y_X=0$. 

\begin{restatable}{lemma}{main}
	\label{lemma:main}
	Let $m$ be a large enough integer, $X \in [m], Y \in \bits^m$ and $R$ be random variables with distribution $p(xyr)$ such that $X \arr R \arr Y$. For any $\gamma$ satisfying $\frac{3}{\log m} \le \gamma^8 \le \frac1{2^{64}}$ define $\CB = \{(x,r)|p(Y_x=0|r) \ge \frac{1+\gamma}2\}$. If 
	\[ \BH(X|R,Y_X=0) \ge (1-\gamma^8)\log m -3  \;\text{    and     }\; \BH(Y|R,Y_X=0) \ge m-\gamma^8 \log m - 3,\]
	then, $p((x,r) \in \CB) \le 64\gamma$.
\end{restatable}

With a little work, one could conclude from the above lemma that if the non-negative rank of $A$ was small, then the probability of the event $X \cap Y = \varnothing$ is at most $\frac12 + O(\gamma)$. Choosing $\gamma$ to be sufficiently small, we can ensure that this probability is much smaller than $\frac12 + \Omega(\rho)$, and so, we derive a contradiction to our initial assumption, that the non-negative rank of the lopsided unique disjointness matrix is small.

To understand Lemma \ref{lemma:main} in more detail, it is worthwhile to consider some simple examples. If it was the case that $\BH(X|R,Y_X=0) \ge \log m - \alpha$ and $\BH(Y|R,Y_X=0) \ge m - \alpha$ where $\alpha \ll 1$, then using Pinsker's inequality, conditioned on the event $Y_X=0$ and most values of $R$, $X$ and $Y$ are close to uniform in statistical distance. Hence, the probability that $x,r$ is such that $p(Y_x=0|r)$ is significantly larger than $\frac12$ is small. Lower bounds on the non-negative rank of the standard unique disjointness matrix (rows and columns indexed by all possible subsets of $[n]$) essentially follow from a variation of this argument since in those cases the entropy loss is small enough so that we can say that the distributions of $X$ and $Y$ (conditioned on $R$ and the event $Y_X =0$) are close to uniform in statistical distance. 

In our case, however, the entropy loss is large enough so that the distributions are quite far from uniform in statistical distance. Let us try to construct an example where the entropy loss is larger. Let $R$ be a random subset of $[m]$ of size $m^{1-\alpha}$ and $T \subseteq R$ be a random subset of size $\alpha\log m$. $X$ will be a uniform index chosen from the subset $R$ while the string $Y$ is chosen uniformly conditioned on the event that the $Y_i=0$ for every $i \in T$. In this case, $R$ is a random variable that breaks the dependency between $X$ and $Y$. Also, we have that $\BH(X|R,Y_X=0)=(1-\alpha)\log m$ and $\BH(Y|R,Y_X=0) = m - \alpha \log m$, so the entropy loss is of the same order as in the assumptions of Lemma \ref{lemma:main}. Here, the event that $x,r$ is such that $p(Y_x=0|r)$ is significantly larger than $\frac12$ occurs only when $x \in T$, and hence, the probability of such $x,r$ is at most $|T|/|R|\le \frac{\alpha \log m}{m^{1-\alpha}}$.

Generalizing the intuition gained from the example given above, the proof of Lemma \ref{lemma:main} proceeds by showing that if the measure of $x,r$ such that $p(Y_x=0|r) \ge \frac{1+\gamma}{2}$ is large, as well as the entropy $\BH(X|R,Y_X=0)$ is large, then for most values of $R$, there is a large set of coordinates of $Y$ that must be very biased given $R$, and hence the entropy $\BH(Y|R,Y_X=0)$ must be small. This intuition is borrowed from the lower bounds on lopsided disjointness in communication complexity \cite{P11, RR15}, but since the setting of non-negative rank is different, the technical details involved for converting this intuition into proof are more challenging here. 

Before moving on, we stress two key points about Lemma \ref{lemma:main}. Firstly, given the assumptions on entropy here, one may hope to say that the probability that $x,r$ is such that $p(Y_x=0|R=r) \le \frac{1-\gamma}2$ must also be small. However, this is not true $-$ the lemma is one-sided and it is fairly easy to construct examples where this is not the case. And secondly, the lopsided structure is crucial in Lemma \ref{lemma:main}. Such a lemma is not true if one considers the case where $Y$ is also a random subset of $[m]$ of size one satisfying $\BH(Y|R,X\cap Y=\varnothing) \ge (1-\gamma^8)\log m-O(1)$ for a constant $\gamma > 0$. So, even though one could hope that the non-negative rank of the small set unique disjointness matrix (where we restrict both rows and columns to be indexed by sets of size at most $k$) is $\approx\binom{n}{k}$, a common information based approach, as used here, will be unable to prove this. 

\subsection{Matching Slack Matrix}
\label{sec:matsketch}

Recalling the connection between extension complexity and non-negative rank of slack matrices, it turns out that to prove Theorem \ref{thm:mat} it is sufficient to prove a lower bound on the non-negative rank of an appropriate slack matrix associated with the matching polytope (see Section \ref{sec:slack}). For the sake of providing intuition, we will work with a slightly simpler slack matrix $S$ in this section. The rows of this slack matrix are indexed by odd cuts (subsets) of $[2n+6]$ of size at most $1/\eps$ and the columns are indexed by perfect matchings in the complete graph on the vertex set $[2n+6]$. The entry corresponding to cut $u$ and perfect matching $m$ is $S_{um}=|\delta(u)\cap m|-1$ where $\delta(u)$ is the set of edges of the complete graph on $[2n+6]$ crossing $u$ (edges with exactly one end point in $u$). The true slack matrix whose non-negative rank we need to bound to prove Theorem \ref{thm:mat} is a noisy version of this slack matrix where a small constant $\beta > 0$ is added to every entry.  

Using a direct-sum argument with appropriate conditioning (dividing vertices into $1/\eps$ chunks of size $O(\eps n)$ each), we can reduce our problem to a question about a random odd cut $U$ and a random perfect matching $M$ in a graph with $2t + 6$ vertices where $t := \eps n$. Furthermore, the non-negative rank decomposition gives a random variable $R$ such that $U \arr R \arr M$ where the size of support of $R$ is $\log \nrank(S)$. Denoting by $q(u,m,r)$ the distribution of $U,M,R$, it turns out that the probability $q(U=u, M=m)$ is proportional to $|\delta(u) \cap m|-1$. In particular, if only one edge of the matching $m$ crosses the cut $u$, then it has probability zero under the distribution $q$.

To describe the high-level idea of the proof we need some notation. We first fix an arbitrary perfect matching $\CA$ in the graph and an arbitrary cut $\CZ$ that cuts all edges of $\CA$. We pick a uniformly random partition $B = (B_0, B_1, \dots, B_t)$ of the set $\CZ$ such that $|B_0|=3$ and $|B_i|=2$ for each $i \in [t]$. We call a cut $U$ \emph{consistent} if $U = B_0 \cup B_j$ for some $j \in [t]$ and note that the size of any such cut is always $5$. We call a perfect matching $M$ \emph{consistent} if $M$ always includes the edges of $\CA$ touching $B_0$ and inside the other blocks $B_j$ for $j \in [t]$, either $M$ includes the edges of $\CA$ or $M$ matches the block $B_j$ to itself and the neighbors of $B_j$ under $\CA$ to itself (see Figure \ref{fig:basicpartition}). 

\begin{figure*}[h!]
    \centering
     \includegraphics[height=1.4in]{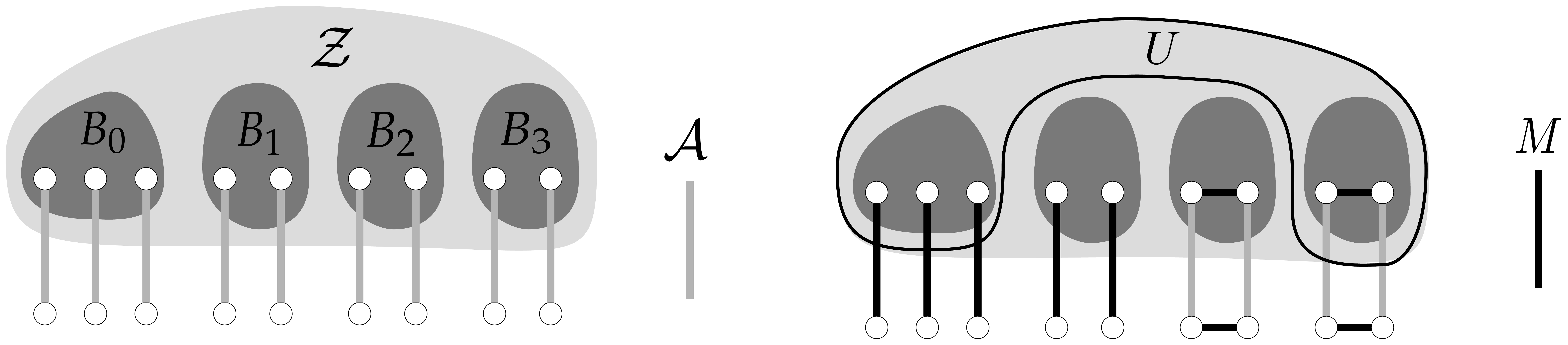}
   	\caption{\small An example of $\CA, \CZ, B$ with $t=3$ (left) and $U,M$ conditioned on $\CD$ (right). Note that $U$ and $M$ are both consistent.}
	\label{fig:basicpartition}
\end{figure*}

Let $\CE$ be the event that $U$ and $M$ are both consistent. Note that given the partition $B$, one can check whether the cut $U$ is consistent without knowing what the matching $M$ is and vice-versa. Hence, as $U - R - M$, it follows that even conditioned on the event $\CE$, $U$ and $M$ are still independent given $R$ and $B$. Furthermore, when the event $\CE$ occurs then either $|\delta(U) \cap M|=3$ (when $U$ does not cut the edges of $M$ apart from those incident on $B_0$) or $|\delta(U) \cap M|=5$ otherwise. 

Let $\CD \subseteq \CE$ denote the event that $U$ and $M$ are consistent and $U$ does not cut the edges of $M$ apart from those incident on $B_0$. Comparing this setup to the case of lopsided unique disjointness, one can see that apart from choosing $B_0$, $U$ corresponds to picking a set of size one among the $t$ blocks $(B_1, \dots, B_t)$, and apart from the edges touching $B_0$, $M$ corresponds to picking a subset of the same $t$ blocks by considering the elements where $M$ crosses the block $B_j$ to be in the set. The event $\CD$ then exactly corresponds to the event that these sets are disjoint. In the ``disjoint'' case, there are exactly $3$ edges of $M$ crossing the cut $U$ where as in the ``intersecting'' case the number of edges of $M$ crossing $U$ is exactly $5$.

As the probability under $q$ is proportional to $|\delta(U) \cap M|-1$, it is not too hard to see that the probability of the event $q(\CD|\CE) = \frac{(3-1)}{(3-1) + (5-1)} = \frac13$. Furthermore, the entropies $\BH(U|B, \CD) = \log t - O(1)$ and $\BH(M|B, \CD) = t - O(1)$. If $\log \nrank(S) \ll \gamma^8 \cdot \frac{n}{t} \log t$ for a small constant $\gamma$, then even after conditioning on $R$ the entropies remain large: $\BH(U|RB, \CD) \ge \log t -\gamma^8 \log t - O(1)$ and $\BH(M|RB, \CD) \ge t - \gamma^8 \log t - O(1)$. 

We want to proceed similarly to the case of lopsided unique disjointness: we want to conclude that the assumptions on entropy imply that $q(\CD|\CE)$ must be much smaller than $\frac13$. In the case of lopsided disjointness, it was enough for us to use Lemma \ref{lemma:main} and bound the contribution to $p(\CD|\CE,R=r,B=b)$ by $\frac12+\gamma$ for most $r, b$. But now as we want to prove that the probability is smaller than $\frac13$, we need to exploit the combinatorial structure of the matching polytope.

This is where the random partition $B$, which is the key new idea introduced by Rothvo{\ss } \cite{R14}, is useful. We are going to argue that only certain kinds of partitions can contribute to the probability of the event $\CD$ otherwise there is a non-zero probability of sampling a cut $U$ and a matching $M$ that satisfies $|\delta(U) \cap M|=1$ and any such pair has probability zero in the distribution $q$, since the corresponding slack matrix entry in $S$ is zero. Then, averaging over all the choices of the random partition $B$, we can show the total contribution of these random partitions to $q(\CD|\CE)$ is indeed less than $\frac13$.

Let us make some simplifying assumptions first. Define $M_j$ to be the edges of $M$ corresponding to block $B_j$. Let us assume that for all values $r,b$ the probability that $M_j = \CA_j$ ($M$ crosses $B_j$) is roughly $\frac12$ conditioned on $\CE,R=r,B=b$ for each $j \in [t]$. Since conditioned on $\CE$ either $M_j=\CA_j$ or $M_j \neq \CA_j$, this implies that 
	\[ q(M_j\neq \CA_j|\CE,R=r,B=b) \approx q(M_j= \CA_j|\CE,R=r,B=b) \approx \frac12.\]
Also assume that $U$ is almost uniform among the $m$ possible cuts. Then, as $U$ and $M$ are independent conditioned on $\CE,R=r,B=b$, we get that 
\[ q(U=B_0\cup B_j,M_j\neq \CA_j|\CE,R=r,B=b) \approx q(U=B_0\cup B_j,M_j = \CA_j|\CE,R=r,B=b).\]

\begin{figure*}[h!]
    \centering
     \includegraphics[height=1in]{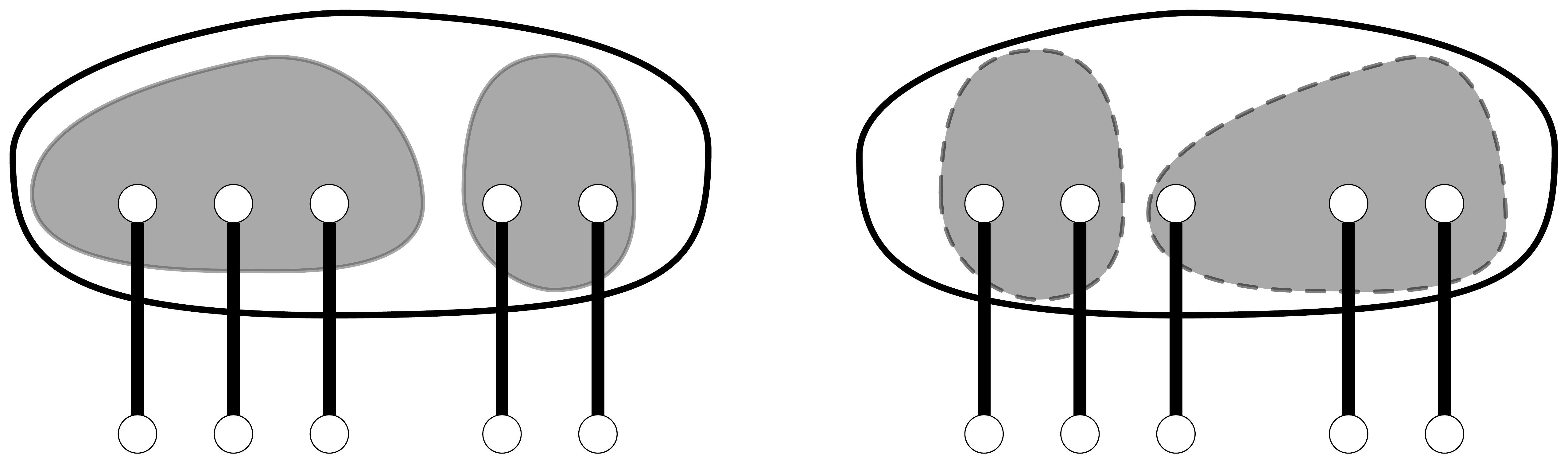}
   	\caption{\small An example of a cut and matching corresponding to the event $U=B_0\cup B_j,M_j = \CA_j$ and an example of two different splits of the cut into $B_0 \cup B_j$. By symmetry any such partition is equi-probable.}
	\label{fig:goodcaseexample}
\end{figure*}

The event on the right hand side above fixes a cut of size $5$ in the graph and a matching that crosses on all the edges. Fix the blocks outside this cut arbitrarily. By symmetry all of the $\binom{5}{3}$ ways of splitting this cut into $B_0$ and $B_j$ are equally likely (see Figure \ref{fig:goodcaseexample}). It turns out that the right hand side above can be non-zero only when the vertices chosen in $B_0$ form a $2$-intersecting family (note that this determines $B_j$ as the parts outside are already fixed) and hence averaging over $b$ and $r$, we get that $q(\CD|\CE) \approx \frac{\Gamma}{10}q(\comp\CD|\CE)$ where $\Gamma$ is a bound on the size of any such family. It turns out that $\Gamma$ is small enough so that we can conclude $q(\CD|\CE) < \frac13$ and derive a contradiction. 

Why must the probability be zero when there are two partitions $b$ and $b'$ (which are same everywhere outside this cut) such that $|b_0 \cap b'_0|=1$? This is because we can choose an appropriate cut $u$ that is consistent with $b'$ (See Figure \ref{fig:errorex}) and has non-zero probability $q(U=u|\CE,R=r,B=b')>0$ and an appropriate matching $m$ that is consistent with $b$ (see Figure \ref{fig:errorex}) and has non-zero probability $q(M=m|\CE,R=r,B=b)>0$. Note that then it must also hold that the cut and the matching has non-zero probability even without conditioning on $\CE$: $q(U=u|R=r)>0$ and $q(M=m|R=r)>0$. But since $U$ and $M$ are independent given $R$ then such a pair would have non-zero probability even though $|\delta(u) \cap m|=1$ and any such pair must have zero probability under $q$. 

\begin{figure*}[h]
    \centering
    \begin{subfigure}[h]{0.8\textwidth}
        \centering
        \includegraphics[height=1in]{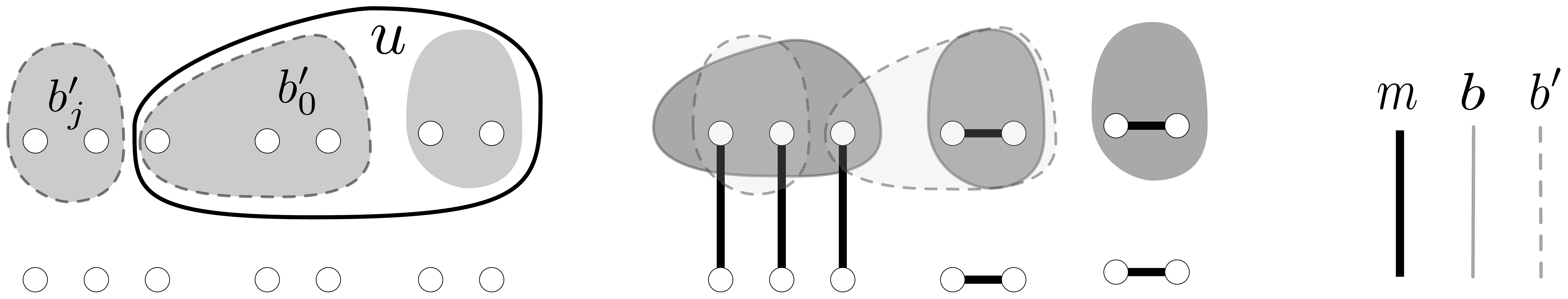}
        \caption{\footnotesize An example of a cut $u$ consistent with $b'$ (left) and a matching $m$ consistent with $b$ (right). Note that $b$ and $b'$ agree on all the blocks except $b_0$ and $b_j$.}
		\label{fig:errorex}
    \end{subfigure}%
	\vspace*{1cm}	
    \begin{subfigure}[h]{0.8\textwidth}
        \centering
        \includegraphics[height=1in]{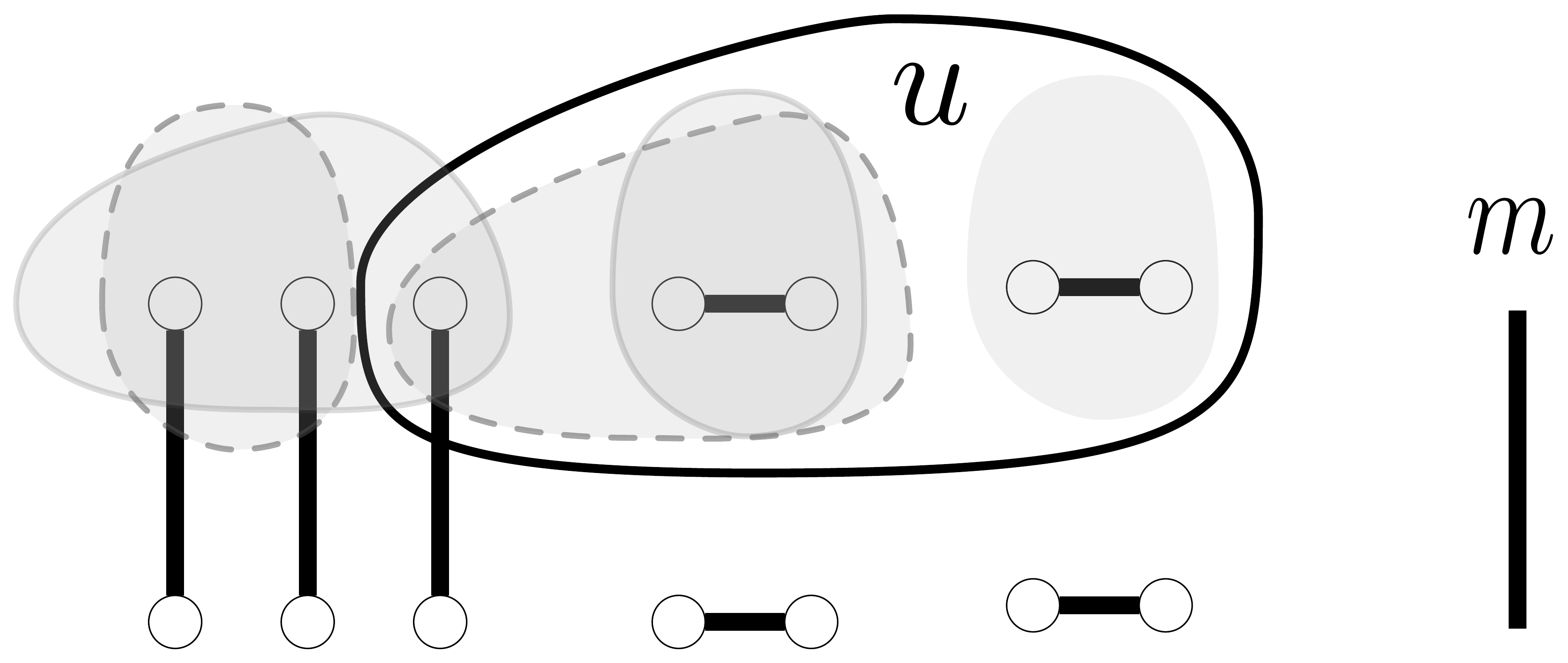}
		\caption{\footnotesize The cut $u$ and matching $m$ satisfy $|\delta(u) \cap m|=1$.}
		\label{fig:errorex1}
    \end{subfigure}
	\label{fig:errorexample}
\end{figure*}

If it was the case that $\BH(U|B, \CD) \ge \log t - \alpha$ and $\BH(M|B, \CD) \ge t - \alpha$ where $\alpha \ll 1$, then one could work with statistical distance as is done in \cite{R14, BP15}, but since in our case the entropy loss is much larger, we use Lemma \ref{lemma:main} in conjunction with the random partition idea. 

As mentioned before, to prove Theorem \ref{thm:mat} we have to work with a noisy version of the slack matrix used above, where the probability of  a pair $U$ and $M$ satisfying $|\delta(U)\cap M|=1$ is not zero, but a small constant. This along with the limitations of Lemma \ref{lemma:main} and the lopsided structure makes translating this intuition into a formal proof considerably more involved.

\section{Slack Matrices and Non-Negative Rank} 
\label{sec:slack}

Consider polytopes $P = \mathsf{conv}\{x_1,\dots, x_s\} \subseteq \BR^d$ and $Q =\{x \in \BR^n ~|~ \langle a_i, x\rangle \le b_i, ~ \forall i \in [t]\} \subseteq \BR^d$ such that $P \subseteq Q$. The \emph{Slack Matrix} $S^{P,Q} \in \BR^{s \times t}$ corresponding to polytopes $P$ and $Q$ is a non-negative matrix defined by $S^{P,Q}_{ij} = b_i - \langle a_i, v_j \rangle$.

The \emph{non-negative rank} of a non-negative matrix $S$ is 
\[ \nrank(S) = \min\{r | \exists U \in \BR_{\ge 0}^{s\times r}, V \in \BR_{\ge 0}^{r\times t} : S=UV\}.\]

Braun, Fiorini, Pokutta and Steurer \cite{BFPS15} showed that to lower bound the size of extended formulations for any polytope sandwiched between inner polytope $P$ and outer polytope $Q$, it suffices to lower bound the non-negative rank of the corresponding slack matrix $S^{P,Q}$.

\begin{thm}[\cite{BFPS15}]
	\label{thm:xc}
	Let $P = \mathsf{conv}\{v_1,\dots,v_s\} \subseteq \BR^d$ and $Q = \{x| \langle a_i,x \rangle \le b_i \;,\forall 1 \le i \le t\} \subseteq \BR^d$ be two polytopes. Then, for any polytope $K \subseteq \BR^d$ satisfying $P \subseteq K \subseteq Q$, $xc(K) \ge \nrank(S^{P,Q}) - 1$.
\end{thm}

\subsection*{Slack Matrix for the Matching Polytope}

Let $0<\beta<1$ be a constant which will be determined by the proof. In proving Theorem \ref{thm:mat}, we may assume without loss of generality that $\frac{c'}{n} \le \eps \le 1$ for a sufficiently large constant $c'$. We call a subset of vertices a \emph{cut} and for a cut $U$, we define $\delta(U)$ to be the set of edges with exactly one end point in $U$ and $E(U)$ to be the set of edges inside $U$.

Consider the following outer polytope $Q'_\eps(n) \subseteq \BR^{\binom{[n]}{2}}$ for matching:
	\begin{align*}
		\ Q'_\eps(n) = \Bigg\{x \in \BR^{\binom{[n]}{2}} ~\Bigg|~&\sum_{e \in \delta(\{i\})} x_e \le1\;,\forall i \in [n];\;x \ge 0;\\
		\		  						     &\sum_{e \in E(U)} x_e \le \frac{|U|+\beta-1}{2} \;,\forall U \subseteq [n],\; |U| \text{ odd },\;|U| \le 1+\frac{\beta}{\eps}.\Bigg\}
	\end{align*}

For any cut $U$ of size at most $1+\frac{\beta}{\eps}$, $(1+\eps)(|U|-1)\le |U| + \beta - 1$, so we have the following.
\begin{prop} $(1+\eps)P_M(n) \subseteq Q'_{\eps}(n)$.
\end{prop}

Hence, to lower bound the extension complexity of approximating polytopes, it suffices to lower bound the non-negative rank of the slack matrix $S^{P_{MAT}, Q'_{\eps}}$. Since there are only $O(n)$ degree constraints, let us restrict ourselves to the submatrix $S$ of $S^{P_{MAT}, Q'_{\eps}}$ given by odd cuts and perfect matchings. Then, the entry corresponding to cut $u$ and perfect matching $m$ is 
\begin{equation}
\ S_{um} = \frac{|u|+ \beta-1}{2} - |E(u)\cap m|=\frac{|\delta(u) \cap m|+\beta-1}{2},
\label{eqn:matslack}
\end{equation}
since for any perfect matching $m$ and odd cut $u$, $|E(u) \cap m|=\frac{|u|-|\delta(u)\cap m|}{2}$. Note that this matrix has $\Theta(1/\eps)$ rows. We prove that

\begin{restatable}{thm}{matrank}
	 \label{thm:matrank}
	 \[ \nrank(S) \ge \binom{n}{{\alpha}/{\eps}}\]
	 where $\frac{c'}{n} \le \eps \le 1$ for a large enough constant $c'=c'(\beta)$ and $0<\alpha<1$ is an absolute constant.
\end{restatable}

Theorems \ref{thm:xc} and \ref{thm:matrank} then give us Theorem \ref{thm:mat}.

\subsection*{Slack Matrix for the Lopsided Correlation Polytope}
\label{sec:acorr}

Consider the polytope $Q = Q(n) = \{x\in \BR^{n\times n}~|~\langle2\diag(a)-aa^T,x\rangle \le 1, a\in \bits^n\}$ where $\diag(a) \in \BR^{n \times n}$ denotes the diagonal matrix which has the vector $a$ on the diagonal and is zero otherwise. It is well-known (see \cite{BP16} for example) that $P_{LCORR} \subseteq Q$ and the slack matrix corresponding to the inner polytope $P_{LCORR}$ and outer polytope $\sigma Q$ is the non-negative matrix $\sigma A$ where $A$ is a lopsided unique disjointness matrix satisfying the conditions in \eqref{eqn:udisj} with parameter $\rho = \frac{1}{\sigma}$. Corollary \ref{cor:acorr} then directly follows from Theorems \ref{thm:udisj} and \ref{thm:xc}. 

\section{Notation and Preliminaries}
\label{sec:preliminaries}

\subsection{Probability Spaces and Variables}

Unless otherwise stated, logarithms in this text are computed base two. We denote by $[n]$ the set $\{1,2,\dotsc, n\}$ and by $\binom{[n]}{k}$ the collection of subsets of $[n]$ of size $k$. Random variables are denoted by capital letters (e.g.\ $A$)
and the values they attain are denoted by lower-case letters (e.g.\ $a$). Events in a probability space will be denoted by calligraphic letters (e.g.\ $\CE$).  For event $\CD$, we use $\ind_\CD$ to denote the corresponding indicator variable, $\comp{\CD}$ to denote the complement event and for another event $\CE$, we use the shorthand $\CD,\CE$ to denote the intersection event $\CD \cap \CE$. Given $a = a_1,a_2, \dotsc, a_n$ (resp. $a=a^1,\dots,a^n$), we write $a_{\leq i}$ (resp. $a^{\leq i}$) to denote $a_1,\dotsc, a_i$ (resp. $a^1,\dots,a^i$). We define $a_{< i}, a_{\ge i}, a_{> i}$ (resp. $a^{<i}, a^{\ge i}, a^{>i}$) similarly.

Given a probability space $p$ and a random variable $A$ in the underlying sample space, we use the notation $p(a)$ to denote both the distribution on the variable $a$, and the number $\BP_p[A = a]$. The meaning will be clear from context. We will often consider multiple probability spaces with the same underlying sample space, so for example $p(a)$ and $q(a)$ will denote the distribution of the random variable $A$ under the probability spaces $p$ and $q$ respectively with the underlying sample space of $p$ and $q$ being the same. We write $p(a|b)$ to denote either the distribution of  $A$ conditioned on the event $B=b$, or the number $\BP[A=a |B=b]$. Given a distribution $p(a,b,c,d)$, we write $p(a,b,c)$ to denote the marginal distribution on the variables $a,b,c$ (or the corresponding probability). We often write $p(ab)$ instead of $p(a,b)$ for conciseness of notation. If $\CE$ is an event, we write $p(\CE)$ to denote its probability according to $p$. 

The support of a random variable $A$ is defined to be the set $\supp(A) := \{a~|~p(a)>0\}$. Given a fixed value $c$, we denote by $\Ex{p(b|c)}{g(a,b,c)} := \sum_{b} p(b|c) \cdot g(a,b,c)$, the expected value of the function $g(a,b,c)$ under the distribution $p(b|c)$. If the probability space $p$ is clear from the context, then we will just write $\Ex{b|c}{g(a,b,c)}$ to denote the expectation. We use $\BE_{a \in \CA}[g(a)]$ to denote the expected value of $g(a)$ under the uniform distribution over the set $\CA$. 

We write $A\arr R\arr B$ to assert that the random variables $A, R, B$ form a Markov chain, or, in other words, $p(arb) = p(r) \cdot p(a|r) \cdot p(b|r)$. In stating the preliminary lemmas and definitions, $p$ is assumed to be the underlying probability space of the random variables being considered.    

{To get familiar with the notation, consider the following example. Let $A \in \bits^2$ be a uniformly distributed random variable in a probability space $p$. Then, $p(a)$ is the uniform distribution on $\bits^2$ and if $a=(0,0)$, $p(a) = 1/4$. Let $A_1$ and $A_2$ denote the first and second bits of $A$, then if $B = A_1 + A_2 \bmod 2$, then when $b=1$, $p(a|b)$ is the uniform distribution on $\{(0,1),(1,0)\}$. If $a = (1,0)$, and $b=1$,  then $p(a|b)=1/2$, and $p(a,b) = 1/4$. If $\CE$ is the event that $A_1=B$, then $p(\CE)=1/2$. Let $q(a)=p(a|\CE)$, then $q(a)$ is the uniform distribution on $\{(0,0),(1,0)\}$ and $q(a_2)$ is the distribution over the sample space $\bits$ which takes the value $0$ with probability $1$.}

\subsection{Entropy and Mutual Information}
\label{sec:prelim_info_theory}

For a discrete random variable $A$, the entropy of $A$ is defined as 
	\[ \BH(A) = \BE_{\cp(a)}\left[\log\frac{1}{\cp(a)}\right]. \]

For any two random variables $A$ and $B$, the entropy of $A$ conditioned on $B$ is defined as $\BH(A|B) = \BE_{\cp(b)}[\BH(A|b)]$. The mutual information between $A$ and $B$ is defined as $\Inf{A}{B} = \BH(A) - \BH(A|B) = \BH(A) - \BH(B|A)$. Similarly, the conditional mutual information is defined as $\Infc{A}{B}{C} = \BH(A|C) - \BH(A|BC)$.

{We shall often work with multiple probability spaces over the same underlying sample space. To avoid confusion, we shall explicitly write $\BH_p(A)$ (and $\Infc[p]{A}{B}{C}$) to specify the probability space $p$ being used for computing the entropy (and mutual information).} 

\subsection{The Binary Entropy Function}

The binary entropy function\footnote{We adopt the convention that $x\log x=0$ at $x=0$.} is defined to be $\sh(x) := -x\log x - (1-x)\log (1-x)$ for $x \in [0,1]$. The function $\sh$ is concave on the interval $[0,1]$ and is decreasing on the interval $[\frac12,1]$.

The proposition below will be quite useful. A proof is given in Appendix \ref{sec:prelimproofs}.

\begin{prop}
\label{prop:binent}
$\sh(x) \le 1 - 2\log e\left(x-\frac12\right)^2$ for all $x \in [0,1]$.
\end{prop}

\subsection{Non-negative Rank and Common Information}

Given a non-negative matrix $A \subset \BR^{\CX \times \CY}$, we can view it as a probability distribution $p(xy)$ as follows:
\[ p(xy) = \frac{A_{xy}}{\sum_{x',y'} A_{x'y'}}.\]
Let $X$ and $Y$ denote random variables with distribution $p(xy)$. Then, we have the following proposition whose proof can be found in Appendix \ref{sec:prelimproofs}.

\begin{prop}[\cite{BP16}]
	\label{prop:nrank}
	There is a random variable $R$ such that $X \arr R \arr Y$ and $|\supp(R)| \le \nrank(A)$.
\end{prop}

One can view the above proposition as saying that to prove a lower bound on the non-negative rank, it suffices to lower bound a well-known information theoretic quantity called the \emph{common information} between $X$ and $Y$. For more details on this interpretation, see \cite{BP16}.

\subsection{Preliminary Information Theory Lemmas}

The proofs of the following basic facts can be found in \cite{CT06}:

\begin{proposition} \label{proposition:infoupper} If $A \in \{0,1\}^\ell$, then $\Inf{A}{B} \leq \ell$.\end{proposition}

\begin{proposition} 
	$\BH(A|B) \le \BH(A)$ where the equality holds if and only if $A$ and $B$ are independent.
\end{proposition}

The above implies that if $A$ and $B$ are independent then $\Inf{A}{B}=0$.

\begin{proposition}[Chain Rule] If $A= A_1,\dotsc, A_n$, then 
	$\BH(A) = \sum_{i=1} \BH(A_i|A_{<i}) \text{ and } \Inf{A_1,\dots, A_n}{B} = \sum_{i=1}^n \Infc{A_i}{B_i}{A_{<i}}.$
\end{proposition}

\begin{prop}[Chernoff Bound]
	\label{lemma:chernoff}
	The number of strings in $\bits^m$ with hamming weight at least $3m/4$ is at most $e^{-m/8}\cdot 2^m$.
\end{prop}

The proof of the following lemma can be found in Appendix \ref{sec:prelimproofs}. 

\begin{lemma}[\cite{BR11}]
	\label{lemma:directsum}
	Let $X = X_1,\dots, X_n$ and $Y = Y_1,\dots, Y_n$ be random variables such that the $n$-tuples $(X_1, Y_1),\dots,(X_n, Y_n)$ are mutually independent. Let $R$ be an arbitrary random variable. Then,
	\begin{align*} 
		\ \sum_{i=1}^n \Infc{X_i}{R}{X_{<i}Y_{\ge i}} &\le \Infc{X}{R}{Y} \text{ and } \sum_{i=1}^n \Infc{Y_i}{R}{X_{\le i}Y_{>i}} \le \Infc{Y}{R}{X}.
	\end{align*}
\end{lemma}

\begin{lemma} Let $X$ be a random variable such that $\BH(X) \ge \log \ell - a$ where $\ell = |\supp(X)|$ and $a \ge 0$. Define $\CS = \lbrace x|p(x)\le \dfrac{2^{(a+1)/\gamma}}{\ell}\rbrace$. Then, $p(\CS) \ge 1 - \gamma$.
\label{lemma:entsupp}
\end{lemma}
\begin{proof}
	Set $b = 2^{(a+1)/\gamma}$. Denoting by $\comp{\CS}$ the complement of $\CS$, we can write
	\begin{align*}
		\ \BH(X) &= \sum_{x \in \CS} p(x)\log\left(\frac{1}{p(x)}\right) + \sum_{x \notin \CS} p(x)\log\left(\frac{1}{p(x)}\right) \\
		\      &\le 	 \sum_{x \in \CS} p(x)\log\left(\frac{1}{p(x)}\right) + p(\comp\CS)\log\left(\frac{\ell}{b}\right) \\
		\	   &\le p(\CS) \log\left(\sum_{x\in \CS} \frac{1}{p(\CS)}\right) + p(\comp\CS)\log\left(\frac{\ell}{b}\right),
	\end{align*}
	where the first inequality follows from the definition of $\CS$ and the second from concavity of the $\log$ function. We can further upper bound 
	\begin{align*}
		\ \BH(X) &\le p(\CS) \log\left(\frac{1}{p(\CS)}\right) + p(\CS) \log|\CS| + p(\comp\CS)\log\left(\frac{\ell}{b}\right) \\
		\        &\le 1 + p(\CS)\log \ell + p(\comp\CS)\log\left(\frac{\ell}{b}\right),	
	\end{align*}
	where we used that $x\log \left(1/x\right) \le 1$ for $0\le x \le 1$. Since $\BH(X) \ge \log \ell - a$, we get that 
	\begin{align*}
		\ \log \ell -a & \le 1 + p(\CS)\log \ell  + (1-p(\CS))\log\left(\frac{\ell}{b}\right),
	\end{align*}
	which gives that $p(\CS) \ge 1 - \frac{a+1}{\log b} = 1-\gamma$.
\end{proof}

\begin{lemma}
	\label{lemma:bias}
	Let $Y \in \str^\ell$ and define $\bias_i(Y) := p(Y_i=0) - p(Y_i=1)$. If there is a set $\CS \subseteq [\ell]$ such that $\BE_{i \in \CS}[\bias_i(Y)] \ge 2\gamma$ where $\gamma > 0$, then $\BHc{}(Y) \le \ell - \gamma^2|\CS|$.
\end{lemma}

\begin{proof}
	We may write
	\[ \BE_{i\in \CS}[\bias_i(Y)] = \BE_{i \in \CS}[p(Y_i=0) - (1-p(Y_i=0))] = 2\BE_{i\in \CS}[p(Y_i=0)] - 1,\]
	which by the assumption implies that $\BE_{i \in \CS}[p(Y_i=0)] \ge \frac12 + {\gamma}$.
	
	We can upper bound
	\[ \BE_{i \in S}[\BH(Y_i)] = \BE_{i \in \CS}[\sh(p(Yi=0))] \le \sh(\BE_{i \in \CS}[p(Y_i=0)]), \] 
	where the last inequality follows from the concavity of the binary entropy function $\sh$. Since $\sh$ is a decreasing function on $\left[\frac12, 1\right]$ and $\sh(\frac12 + x) \le 1 - 2\log e\cdot x^2 \le 1 - x^2,$
	\[\BE_{i \in \CS}[\BH(Y_i)] \le \sh\left(\frac12 + {\gamma}\right) \le 1 - \gamma^2.\] 

	Denoting by $\comp{\CS}$ the complement of $\CS$ and applying the chain rule we get: 
	\begin{align*}
	\	\BH(Y) &\le \BH(Y_{\overline \CS}) + \BH(Y_\CS) \le \BH(Y_{\overline \CS}) + \sum_{i \in \CS}\BH(Y_i) \\
	\	     &\le (\ell - |\CS|) + |\CS|(1-\gamma^2) = \ell - \gamma^2|\CS|.
	\end{align*}
\end{proof}

\begin{lemma}[Averaging Lemma]
	\label{lemma:avg}
	Let $A$ be a bounded random variable such that $\BE[A] \ge \alpha$. For any $\beta < \alpha$, let $\CS = \{ a ~|~ A(a) \ge \beta \}$. Then, $p(\CS) \ge \frac{\alpha-\beta}{m-\beta}$ where $m = \max_a\{A(a)\}$.
\end{lemma}

\begin{proof}
	We have 
	\begin{align*}
	\ \BE[A] &= \sum_{a \in \CS} p(a) A(a) + \sum_{a\notin \CS}p(a) A(a) \le p(\CS) m + (1-p(\CS))\beta.
 \end{align*}
	which gives us that $(m-\beta)p(\CS) \ge \BE[A] - \beta \ge \alpha - \beta$.
\end{proof}

\subsection{Intersecting Families}

The following lemma will be crucial for the analysis. It is a special case of the Erd\H{o}s-Ko-Rado Theorem (see \cite{W84}) which says that when $n \ge 6$, then the size of any family of $\binom{[n]}{3}$ that intersects in two elements is at most $n-2$ (Lemma \ref{lemma:intfam} follows from the case $n=6$). We give a self-contained proof in Appendix \ref{sec:prelimproofs}.

\begin{lemma}
	\label{lemma:intfam}
	Let $\FF \subseteq \binom{[5]}{3}$ be a family of subsets such that any two sets in $\FF$ intersect in two elements. Then, $|\FF|\le 4$.
\end{lemma}

\section{Lower Bounds on Non-negative Rank}
\label{sec:lowerbounds}

Let us recall the main technical lemma which we use to derive a lower bound on the non-negative rank of the lopsided unique disjointness matrix as well as the matching slack matrix.

\main*

We will prove the above lemma in Section \ref{sec:mainlemmaproof}. First we use it to derive non-negative rank lower bounds.

\subsection{Non-negative Rank of Lopsided Unique Disjointness}

In this section we prove Theorem \ref{thm:udisj}.

\udisj*

It will be convenient to assume that $n$ is divisible by $k$ and $\frac{n}k$ is a large enough integer. We split the universe $[n]$ into blocks of size $\frac{n}k$ where $\{\frac{n}k(i-1)+1,\dots, \frac{n}{k}i\}$ is the $i^\text{th}$ block for every $i \in [k]$.  

Define a distribution on $X \in \binom{[n]}{k}$ and $Y \in \bits^n$ given by \[q(xy) = \frac{A_{xy}}{\sum_{x',y'} A_{x'y'}},\]
where $A$ is the lopsided unique disjointness matrix defined in \eqref{eqn:udisj} with parameter $\rho$ and we view $y \in \bits^n$ as the indicator vector for a subset of $[n]$. Let $X^i$ denote the intersection of $X$ with the elements in the $i^{\text{th}}$ block, and let $Y^i$ be the projection of $Y$ onto the coordinates in the $i^{\text{th}}$ block. For every $x \in [\frac{n}k]$, we will use the notation $Y^i_x$ to denote the $x^{\text{th}}$ coordinate of $Y_i$, in other words $Y_{\frac{n}{k}(i-1)+x}$. 

Let $\CD$ denote the event that $x$ and $y$ are disjoint and $x^i$ has exactly one element for every $i$.  
\begin{lemma} 
	\label{lemma:base}
	Let $R$ be any random variable satisfying $X \arr R \arr Y$. Then, for every $i \in [k]$ it holds that 
	\begin{align*}
		\ \Infc[q]{R}{X^i}{X^{<i}Y^{\ge i}\CD}  +  \Infc[q]{R}{Y^i}{X^{\le i}Y^{>i}\CD}  &\ge \left(\frac{\rho}{1000}\right)^8 \log \left(\frac{n}{k}\right).
	\end{align*}
\end{lemma}

Using Lemma \ref{lemma:directsum} and Proposition \ref{prop:nrank} together with the above, we have
\[2\log \nrank(M) \ge \sum_{i=1}^k\left(\Infc[q]{R}{X^i}{X^{< i}Y^{\ge i}\CD}  +  \Infc[q]{R}{Y^i}{X^{\le i}{Y^{>i}}\CD}\right) \ge \left(\frac{\rho}{1000}\right)^8 k \log \left(\frac{n}{k}\right) ,\]
which proves Theorem \ref{thm:udisj} as $\binom{n}{s} = 2^{\Theta\left(s\log\left(\frac{n}{s}\right)\right)}$.

\begin{proof}[Proof of Lemma \ref{lemma:base}]
	Fix $i$ as in the statement of the lemma. Let $\Eps$ be the event that $X^i$ has exactly one element for every $i$ and $X \cap Y$ is a subset of the $i^{\text{th}}$ block. Note that $\CD \subset \CE$. 
	
Writing $W = X^{<i}Y^{>i}$, we will prove that for any fixed value $w$ attained by $W$, we have that 
\begin{equation}
	\label{eq:inf}
	\ \Infc[q]{R}{X^i}{wY^i\CD}  +  \Infc[q]{R}{Y^i}{wX^i\CD}  \ge \left(\frac{\rho}{1000}\right)^8 \log \left(\frac{n}k\right), 
\end{equation}
and the proof is completed by averaging over $w$. 

Note that after fixing $w$ and $r$, $X$ and $Y$ are independent and they can be checked separately to verify that the event $\CE$ occurs as for every block $j\neq i$ either $X^j$ or $Y^j$ is fixed given $w$. It follows that the distribution $p(xyr) := q(xyr|w\CE)$ satisfies $p(xy|r)=p(x|r)p(y|r)$. Furthermore under the distribution $p$, $\CD$ is equivalent to the event that $Y^i_{X^i}=0$. We can compute $p(\CD) \ge \frac{1}{1+(1-\rho)} = \frac{1}{2-\rho}$ since the matrix entries are given by
	\begin{equation}
		A_{xy} = \begin{cases} 1 \text{ when } (x, y) \in \supp(p(xy|\CD)),\\
					\le 1-\rho \text{ when } (x,y) \in \supp(p(xy|\comp{\CD})),
\end{cases}
	\label{eqn:disjslack}
\end{equation}
and the number of entries is the same in both cases.

For the sake of contradiction assume that \eqref{eq:inf} does not hold. Then, we are going to show that $p(\CD)$ must be significantly smaller than what we computed above. Define $\CB = \{(x^i,r)~|~p(Y^i_{x^i}=0|r)\ge \frac12 + \frac{\rho}{2000}\}$. We can upper bound $p(\CD)$ as follows:
\begin{align*}
	\ p(\CD) &= \sum_{x^i,r} p(x^i,r)p(Y^i_{x^i}=0|r,X^i=x^i) = \sum_{x^i,r} p(x^i,r)p(Y^i_{x^i}=0|r)  \\
	\		 &\le \sum_{(x^i,r) \notin \CB} p(x^i,r)\left(\frac12+\frac{\rho}{2000}\right) + \sum_{(x^i,r) \in \CB} p(x^i,r)p(Y^i_{x^i}=0|r) \\
	\		 &\le \left(\frac12+ \frac{\rho}{2000}\right) + p((x^i,r) \in \CB),
\end{align*}
where the second equality follows since $p(xy|r)$ is product.

We will show that
\begin{claim}
	\label{claim:badxr}
	$p((x^i,r) \in \CB)) \le \frac{64\rho}{1000}$.
\end{claim}

This implies $p(\CD) \le \frac12 + \frac{\rho}{2000} + \frac{64\rho}{1000} < \frac{1}{2-\rho}$, which contradicts the fact that $p(\CD)\ge\frac{1}{2-\rho}$. This finishes the proof. Next we prove Claim \ref{claim:badxr}.  
\end{proof}

\begin{proof}[Proof of Claim \ref{claim:badxr}]
	Set $t:=\frac{n}{k}$. Conditioned on $x\CD$, every coordinate of $y^i$ other than $x^i$ is uniform and hence $\BHc{p}(Y^i|X^i\CD) \ge t - 1$. Similarly, conditioned on $y\CD$, $X^i$ is uniform on the coordinates of $y^i$ that are zero. We may compute from \eqref{eqn:disjslack} that $p(y^i|\CD) = \frac{\Delta(y^i)}{t2^{t-1}}$ where $\Delta(y^i)$ is the number of zeros in $y^i$. The probability of any string $y^i$ with less than $t/4$ zeros is at most $\frac{1}{2^{t+1}}$ and using Proposition \ref{lemma:chernoff} their total measure under the distribution $p(y^i|\CD)$ can be bounded by $\frac{e^{-t/8}}2$. Hence, $\BHc{p}(X^i|Y^i\CD) \ge (1-\frac{e^{-t/8}}2) \log \left(\frac{t}{4}\right) \ge \log t - 3$.

	As $p(xyr)=q(xyr|w\CE)$, if \eqref{eq:inf} is not true, then 
	\[ \Infc[p]{R}{X^i}{Y^i\CD} = \BHc{p}(X^i|Y^i\CD) - \BHc{p}(X^i|RY^i\CD) \le \left(\frac{\rho}{1000}\right)^8 \log t,\]
	and a similar statement is obtained by writing $\Infc[p]{R}{Y^i}{X^i\CD}$ in terms of entropy. It follows that
\begin{align*}
	\ \BHc{p}(X^i|RY^i\CD) &\ge \left(1-\left(\frac{\rho}{1000}\right)^8\right) \log t - 3, \text{~and~} \BHc{p}(Y^i|RX^i\CD) \ge t - \left(\frac{\rho}{1000}\right)^8 \log t - 1. 
\end{align*}

Since entropy can only decrease under conditioning and $\left(\frac{\rho}{1000}\right)^8 \ge \frac{3}{\log t}$, $X^i, Y^i$ and $R$ satisfy the conditions of Lemma \ref{lemma:main} and the claim follows.
\end{proof}

\subsection{Non-negative Rank of the Matching Slack Matrix}

Let us recall the definition of the matching slack matrix $S$ from \eqref{eqn:matslack}. The entry $S_{um}$ corresponding to cut $u$ and perfect matching $m$ is 
\begin{equation*}
	\ S_{um} = \frac{|\delta(u) \cap m|+\beta-1}2,
\end{equation*}
where $0<\beta<1$ is a constant that will be determined by the proof and $\delta(u)$ is the set of edges crossing $u$. In this section we prove Theorem \ref{thm:matrank}. 

\matrank*

For convenience we assume that $\beta/\eps$ is an integer and we work with graphs on $4n+6$ vertices where $n$ is divisible by $\beta/\eps$. Set $t = \frac{n}{\beta/\eps}$ and note that $t \ge c$ for a large constant $c=c(\beta)$.

Using the slack matrix $S$, define a distribution on cuts $U$ and perfect matchings $M$ given by \[q(um) = \frac{S_{um}}{\sum_{u',m'} S_{u'm'}}.\]

Fix an arbitrary perfect matching $\CA$ and an arbitrary cut $\CZ$ that cuts all edges of $\CA$. Let $B=(B_0, B_1,\dots,B_n)$ be a uniformly random partition of the set $\CZ$ into blocks such that $|B_0|=3$ and $|B_j|=2$ for each $j\in[n]$. For $i\in[\frac{n}t]$, we call $B^i = \{B_{t(i-1)+1},\dots, B_{ti}\}$ the $i^{\text{th}}$ \emph{chunk} and for $j \in [t]$ use $B^i_j = B_{t(i-1)+j}$ to denote the $j^{\text{th}}$ block of the chunk $B^i$. Let $\CA_j$ denote the edges of $\CA$ that touch $B_j$. 

We say that the cut $U$ is \emph{consistent} with $B$ if $B_0 \subseteq U \subseteq \CZ$ and for every $i \in [\frac{n}t]$, $U^i := U \cap (B^i_1 \cup \dots \cup B^i_{t})$ equals $B^i_j$ for some $j \in [t]$. We say that $M$ is \emph{consistent} with $B$ if $\CA_0 \subseteq M$ and for each $j \in [n]$, either $\CA_j \subseteq M$ or $M$ matches $B_j$ to itself and matches the neighbors of $B_j$ under $\CA_j$ to themselves. 

\begin{figure*}[h!]
    \centering
     \includegraphics[height=2.5in]{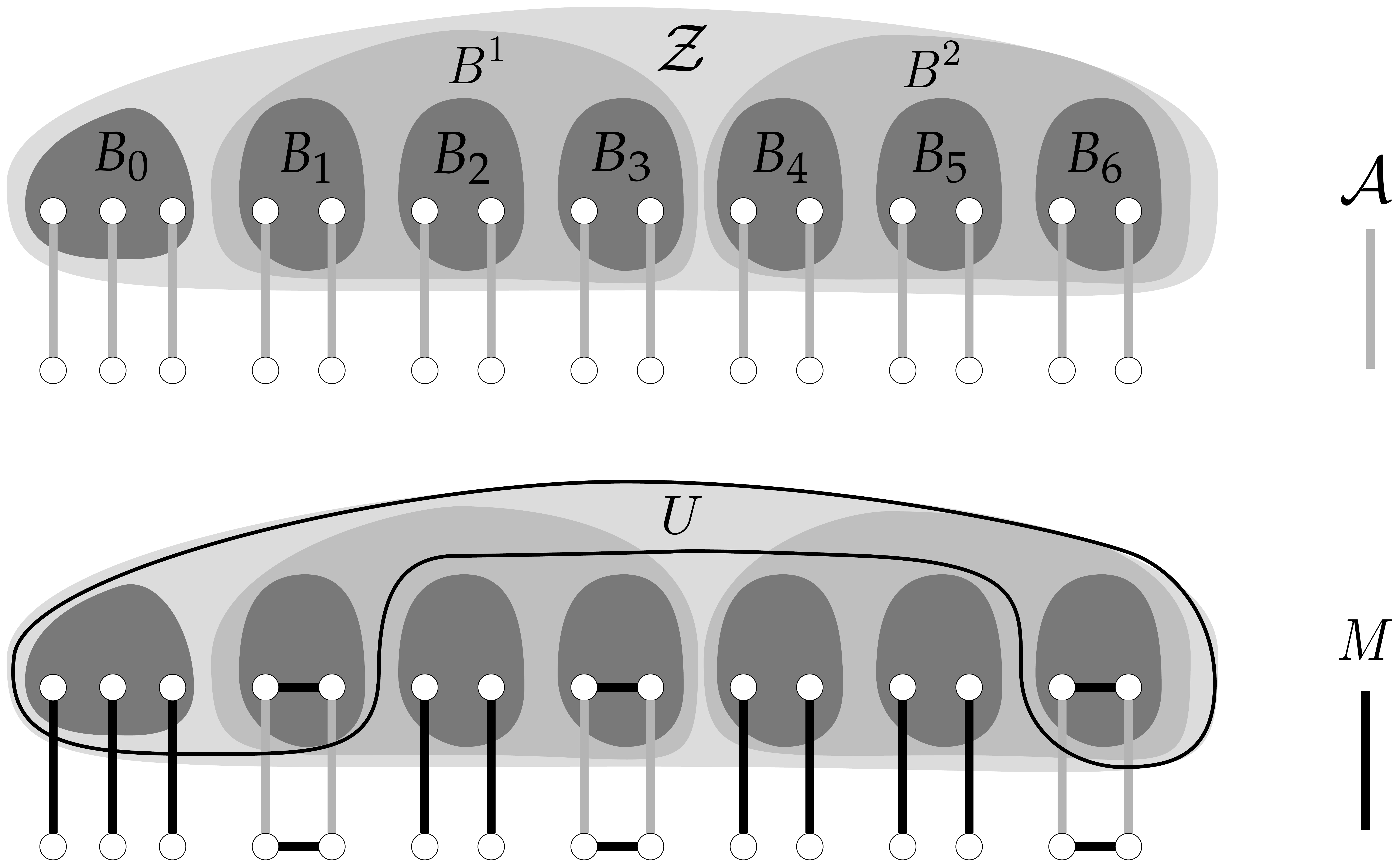}
   	\caption{\small Top to bottom: An example of $\CA, \CZ, B$ with $n=6$ and $t=3$; $U,M$ conditioned on $\CD$.}
	\label{fig:partition}
\end{figure*}

For $i \in [\frac{n}t]$, we write $M^i$ to denote the edges of $M$ contained in the vertices of $B^i$ and neighbors of $B^i$ under $\CA$. We write $M^i_j$ (and $\CA^i_j$) to denote the edges of $M^i$ (and $\CA$) corresponding to the vertices of $B^i_j$ and its neighbors under $\CA$.   

Let $U$ and $M$ be sampled from $q$ and let $\CD$ denote the event that $U$ and $M$ are consistent with $B$ and for every $i \in [\frac{n}t]$, $U^i$ does not cut the edges of $M^i$. 

\begin{lemma} 
	\label{lemma:basemat}
	Let $R$ be any random variable satisfying $U \arr R \arr M$. Then, for every $i \in [\frac{n}{t}]$ it holds that 
	\begin{align*}
		\ \Infc[q]{R}{U^i}{BU^{<i}M^{\ge i}\CD}  +  \Infc[q]{R}{M^i}{BU^{\le i}M^{>i}\CD}  &= \Omega(\log t).
	\end{align*}
\end{lemma}

Using Lemma \ref{lemma:directsum} and Proposition \ref{prop:nrank} together with the above, we have
\[2\log \nrank(S) \ge \sum_{i=1}^{n/t}\left(\Infc[q]{R}{U^i}{BU^{< i}M^{\ge i}\CD}  +  \Infc[q]{R}{M^i}{BU^{\le i}{U^{>i}}\CD}\right) = \Omega\left(\frac{n}t \log t\right) = \Omega\left(\frac{1}{\eps}\log(\eps n)\right),\]
which proves Theorem \ref{thm:matrank} as $\binom{n}{s} = 2^{\Theta\left(s\log\left(\frac{n}{s}\right)\right)}$.

\begin{proof}[Proof of Lemma \ref{lemma:basemat}]
	Fix a value of $i$ as in the statement of the lemma. Let $\Eps$ be the event that $U,M$ are consistent with $B$ and for each $j\neq i$, the edges of $M^j$ are not cut by $U^j$. Note that $\CD \subset \Eps$, but under $\Eps$ edges of $M^i$ may be cut by $U^i$. 

For any partition $b$ and for any $(u,m) \in \supp(p(um|b\CE))$, the weight of the slack matrix entries are given by
\begin{equation}
	S_{um} = \begin{cases} \frac{2+\beta}2 \text{ when } (u, m) \in \CD,\\
	 					\frac{4+\beta}2 \text{ when } (u,m) \in \comp\CD,\\
		   \end{cases}
	\label{eqn:slack}
\end{equation}
and note that the number of entries in both cases above is exactly $t2^{t-1}$.

Writing $W = U^{<i}M^{>i}B^{-i}$, we will prove that for any fixed value $w$ attained by $W$, the following holds 
\begin{equation}
	\label{eq:matinf}
	\ \Infc[q]{R}{U^i}{M^iBw\CD}  +  \Infc[q]{R}{M^i}{U^iBw\CD}  \ge \beta^8 \log t,
\end{equation}
and the proof is completed by averaging over $w$. Observe that the partition of the $i^{\text{th}}$ chunk $B^i$ and $B_0$ is still a random variable even after fixing $w$.

After fixing $wrb$, $U$ and $M$ are independent and they can be checked separately to verify that the event $\CE$ occurs as for every block $j\neq i$ either $U^j$ or $M^j$ is fixed given $wb$. It follows that the distribution $p(umbr)$ defined as $p(umbr) = q(umbr|w\CE)$ satisfies $p(um|rb) = p(u|rb)p(m|rb)$. A direct computation using \eqref{eqn:slack} then shows that \[p(\CD) = \frac{2+\beta}{(2+\beta)+(4+\beta)} = \frac{2+\beta}{6+2\beta} \ge \frac{1}{3}.\] 

For the sake of contradiction, assume that \eqref{eq:matinf} does not hold. Then, we are going to argue that $p(\CD)$ must be significantly smaller than $\frac13$ reminiscent to the proof of Lemma \ref{lemma:base}.
 
Define the random variable $J \in [t]$ to be $J=j$ if $U^i=B^i_j$ and note that the distribution $p(jm|rb)$ is product. Furthermore, note that under the distribution $p$, the event $\CD$ is equivalent to $M^i_J \neq \CA^i_J$.

For any $(j,r,b) \in \supp(p(jrb))$, we define the set of blocks correlated with $j$ as 
\[\CJ_{jrb} = \left\{k ~\bigg|~k\in [t], k \neq j, p(M^i_k \neq \CA^i_k|rb,M^i_j\neq \CA^i_j) \notin \left(\frac14,\frac34\right) \right\}.\] 
Define sets:
\begin{align*}
	\ \CS_1 &= \left\{(j,r,b)~\bigg|~ p(M^i_j \neq \CA^i_j|rb) \le \frac{1}{30}\right\}, &\CS_2 = \left\{(j,r,b)~\bigg|~p(M^i_j \neq \CA^i_j|rb) \ge \frac{1+\beta}{2}\right\},  \\
	\   \CS_3 &= \left\{(j,r,b)~\bigg|~p(j|rb\CD) \ge \frac{2^{1/\beta^2}}{t^{1-\beta^2}}\right\}, &\CS_4 = \left\{(j,r,b)~\bigg|~ |\CJ_{jrb}| \ge \frac{\beta}{2} \cdot \frac{t^{1-\beta^2}}{2^{1/\beta^2}}-1\right\}.
\end{align*}

For brevity, when $\CS \subseteq [t]\times\supp(p(rb))$ we write $p(\CS)$ to denote the probability of the event $(j,r,b) \in \CS$. Using Lemma \ref{lemma:main} and other entropy related arguments, we will be able to show that the contribution of $\CS_1, \CS_2, \CS_3$ and $\CS_4$ to $p(\CD)$ is small.

\begin{claim}
	\label{claim:error}
	\begin{enumerate*}[label={(\alph*)}]
		\item $p(\CS_1,\CD) \le \frac{1}{30}$,
		\item $p(\CS_2,\CD) \le 64\beta$,
		\item $p(\CS_3,\CD) \le 3\beta^2$,
		\item $p(\CS_4,\CD) \le \beta$.
	\end{enumerate*}
\end{claim}

Hence, denoting by $\CG$ the complement of the event $\CS_1 \cup \CS_2 \cup \CS_3 \cup \CS_4$, we can bound 
\begin{align*}
	\ p(\CD) &\le  p(\CG,\CD) + \frac{1}{30} + 64 \beta + 3\beta^2 + \beta \le p(\CG,\CD) + \frac{1}{30} + 2^7\beta \stepcounter{equation}\tag{\theequation}\label{eqn:upper}.
\end{align*}

To bound the contribution of $\CG$, we need to use the combinatorial structure of matchings. For this, we further split $\CG$ into two events. We say $(j,r,b)\in \CG_2$ if $(j,r,b) \in \CG$ and for all partitions $b'$ such that $(j,r,b')\in \CG$ and $b'$ agrees with $b$ on all the blocks except $b_0$ and $b^i_j$, it holds that $|b_0 \cap b'_0|\ge 2$. We define $\CG_1 = \CG\setminus\CG_2$. Note that by definition if $(j,r,b)\in \CG_1$, then there exists another partition $b'$ such that $(j,r,b')\in \CG_1$ and $b'$ agrees with $b$ on all the blocks except $b_0$ and $b^i_j$. 

As discussed in Section \ref{sec:matsketch}, most of the contribution comes from the event $\CG_2$ which we can bound by

\begin{claim}
	\label{claim:good}
	\[ p(\CG_2,\CD) \le \frac{4}{10} p(\comp\CD) + \beta = \frac{4}{10}\cdot\frac{4+\beta}{6+2\beta} + \beta.\] 
\end{claim}

Next we want to relate the contribution of $\CG_1$ to the probability of the event $|\delta(U) \cap M|=1$ under the distribution $q$. Since, this probability is not zero under the slack matrix $S$, we will need some more notation to bound it more carefully. We call a matching $M$ to be \emph{bad} for $B$ if $M$ includes exactly one edge crossing $B_0$ which is an edge of $\CA_0$ and for other blocks either $\CA_j \subseteq M$ or $M$ matches $B_j$ to itself and matches the neighbors of $B_j$ under $\CA_j$ to themselves. See Figure \ref{fig:badmat}.

\begin{figure*}[!h]
		\centering
        \includegraphics[height=1.2in]{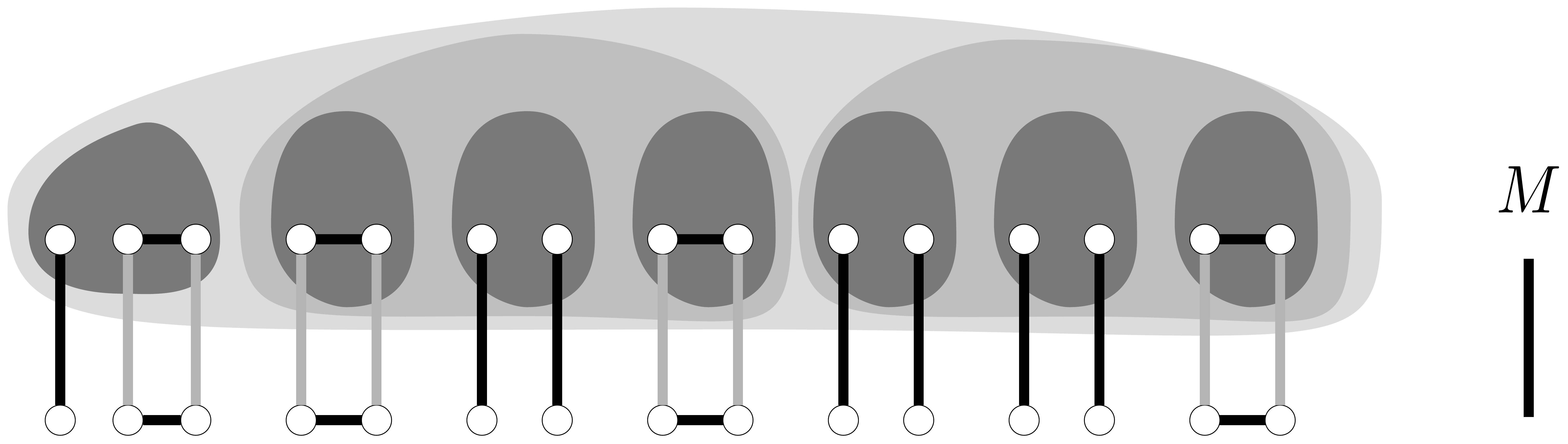}
        \caption{\footnotesize A bad matching $M$}
	\label{fig:badmat}
\end{figure*}

Let $\CF$ denote the event that $U$ is consistent with $B$, $M$ is \emph{bad} for $B$ and for each $j\neq i$, the edges of $M^j$ are not cut by $U^j$. For any partition $b$, the entries of the slack matrix when $(u,m) \in \supp(q(um|b\CF)$ are given by
\begin{equation}
	S_{um} = \begin{cases} \frac{2+\beta}2 \text{ when } |\delta(u) \cap m|=3\\
			  \frac{\beta}2 \text{ when } |\delta(u) \cap m|=1.
\end{cases}
	\label{eqn:slackbad}
\end{equation}
Note that the number of entries in the two cases above is exactly $3t2^{t-1}$ since for every fixing of all the edges of $M$ except $M_0$, there are exactly three matchings that are bad.

With the above, we will be able to show that 
\begin{claim}
	\label{claim:bad}
	\[ p(\CG_1,\CD) \le 48000 q(|\delta(U)\cap M|=1~|~\CE\cup\CF) + \beta \le 2^{14}\beta.\]
\end{claim}

Plugging the above claims in \eqref{eqn:upper} and choosing $\beta$ to be a sufficiently small constant, we derive that 
\[ p(\CD) \le \frac{4}{10}\cdot\frac{4+\beta}{6+2\beta} + \frac{1}{30} + 2^7\beta + 2^{14}\beta \le \frac{9}{30} + 2^{15}\beta,\]
which is a contradiction as $p(\CD) = \frac{2+\beta}{6+2\beta} \ge \frac{1}{3}$. This finishes the proof. Next we turn to proving the claims. 
\end{proof}

\begin{proof}[Proof of Claim \ref{claim:good}]

For any $(j,r,b)\in \CG$, we have $p(M^i_j \neq \CA^i_j|rb) \le p(M^i_j = \CA^i_j|rb) + \beta$ since  
\[p(M^i_j \neq \CA^i_j|rb) - p(M^i_j = \CA^i_j|rb) = 2p(M^i_j \neq \CA^i_j|rb) - 1 \le 2\cdot\frac{1+\beta}{2} - 1 = \beta.\]

Recall that $\CD$ is equivalent to the event $M^i_J \neq \CA^i_J$ under $p$. As $\CG_2 \subset \CG$ and $p(jm|rb)$ is product, 
\begin{align*}
	\ p(\CD, \CG_2) &= p(M^i_J \neq \CA^i_J, \CG_2) = \sum_{(j,r,b) \in \CG_2} p(jrb) p(M^i_j \neq \CA^i_j|rb) \\
	\			&\le \sum_{(j,r,b) \in \CG_2} p(jrb) p(M^i_j = \CA^i_j|rb) + \beta = p(M^i_J = \CA^i_J, \CG_2) +\beta.
\end{align*}

Conditioned on the event $J=j, M^i_j = \CA^i_j, R=r$, by symmetry the partition $B^i$ is uniform among all partitions which agree on all the blocks except $B_0$ and $B^i_j$ (recall Figure \ref{fig:goodcaseexample}). For every fixing of the blocks outside $B_0$ and $B^i_j$, there are $\binom53$ such partitions out of which at most $4$ can be in $\CG_2$ by Lemma \ref{lemma:intfam}. In particular, this means that for any $(j,r)$ we can bound 
\[ p(\CG_2|~M^i_j\neq \CA^i_j,J=j,r) = \sum_{b:(j,r,b) \in \CG_2} p(b|M^i_j = \CA^i_j,J=j,r) \le \frac{4}{\binom{5}{3}} = \frac{4}{10} \] 

Averaging over $j$ and $r$, we get that $p(\CG_2|M^i_J = \CA^i_J) \le \frac{4}{10}$. Finally, 
\begin{align*}
\	p(\CG_2,\CD) &\le p(M^i_J = \CA^i_J, \CG_2) + \beta =  p(M^i_J = \CA^i_J) p(\CG_2|M^i_J = \CA^i_J) \\
\				&\le \frac{4}{10} p(M^i_J = \CA^i_J) + \beta = \frac{4}{10}p(\comp\CD) + \beta.
\end{align*}
\end{proof}

\begin{proof}[Proof of Claim \ref{claim:bad}]
	For every $(r,b)$ fix an $\ell$ arbitrarily such that $(\ell,r,b) \in \CG_1$ (if such an $\ell$ does not exist then the contribution of $r,b$ to $p(\CG_1,\CD)$ is zero anyway). Then by definition, there exists a partition $b'$ such that $(\ell,r,b')\in \CG_1$, $b$ and $b'$ agree on all the blocks except $b_0 \cup b^i_\ell$ and $|b_0 \cap b'_0|=1$. Note that $b'$ is determined by $r,b$. Define $\CT_{rb} := \CJ_{\ell rb} \cup \CJ_{\ell rb'} \cup \{\ell\}$ and let $\CT$ denote the event that $J \in \CT_{rb}$. We will show that when $(j,r,b) \in \CG_1$ and $j \notin \CT_{rb}$:
	\begin{equation}
		\ p(J=j,M^i_j \neq \CA^i_j,r,b) \le 48000 q(J=j, M^i_j \neq \CA^i_j, M \text{ is bad},r,b|\CE\cup\CF).
	\label{eqn:error}
	\end{equation}

The above implies that
	\begin{align*}
	\ p(\CG_1, \comp\CT, M^i_J \neq \CA^i_J) &\le 48000 q(\CG_1, \comp\CT, M^i_J \neq \CA^i_J, M \text{ is bad}|\CE\cup\CF)\\
	\								   &\le 48000 q(M^i_J \neq \CA^i_J, M \text{ is bad}|\CE\cup\CF). \stepcounter{equation}\tag{\theequation}\label{eqn:term1}
	\end{align*}

	Furthermore, since both $(\ell,r,b)$ and $(\ell,r,b')$ are in $\CG_1 \subseteq \CG$, $|\CT_{rb}| \le {\beta}\cdot\frac{t^{1-\beta^2}}{2^{1/\beta^2}}-1 \le {\beta}\cdot\frac{t^{1-\beta^2}}{2^{1/\beta^2}}$ and also $p(j|rb, \CD) \le \frac{2^{1/\beta^2}}{t^{1-\beta^2}}$ when $(j,r,b) \in \CG_1$. So, we have $p(\CT,\CG_1|\CD) \le \beta.$

	When $M^i_J \neq \CA^i_J$ and $M$ is bad, then the cut $U^i = B_0 \cup B^i_J$ satisfies $|\delta(U) \cap M|=1$ (see Figure \ref{fig:matb}). So, from \eqref{eqn:term1} and the above, we get 
\begin{align*}
	\ p(\CG_1, \CD) &= p(\CG_1, M^i_J \neq \CA^i_J) \le p(\CG_1, \comp\CT, M^i_J \neq \CA^i_J) + p(\CG_1, \CT|\CD) \\
	\			   &\le 48000 q(M^i_J \neq \CA^i_J, M \text{ is bad}|\CE\cup\CF) + \beta \le 48000q(|\delta(U)\cap M|=1~|~\CE\cup\CF) +\beta.
\end{align*}

Using \eqref{eqn:slack} and \eqref{eqn:slackbad}, the first term on the right hand side is $\frac{3\beta}{(2+\beta)+(4+\beta)+3(2+\beta)+3\beta} \le \frac{\beta}{4}$ and hence $p(\CG_1,\CD)\le {2^{14}\beta}$. 

All that remains is to prove \eqref{eqn:error}. Since $(\ell,r,b) \in \CG$ and $j \notin \CT_{rb}$, $p(M^i_\ell \neq \CA^i_\ell|rb) \ge \frac{1}{30}$ and $p(M^i_j \neq \CA^i_j|rb, M^i_\ell \neq \CA^i_\ell) \ge \frac14$, so $p(M^i_j \neq \CA^i_j, M^i_\ell \neq \CA^i_\ell|rb) \ge \frac1{120}.$ Since probability is always less than one, trivially 
\[ p(M^i_j \neq \CA^i_j|rb) \le 120p(M^i_j \neq \CA^i_j, M^i_\ell \neq \CA^i_\ell|rb).\]

Using the above and the fact that $p(jm|rb)$ is product, we get 
\begin{align}
	\  p(J=j, M^i_j \neq \CA^i_j,r,b) &\le 120p(J=j, M^i_j \neq \CA^i_j, M^i_\ell \neq \CA^i_\ell,r,b) 
	\label{eqn:midstep}
\end{align}

Next we relate the probability on the right hand side above to the probability under $q$ conditioned on the event $\CE \cup \CF$ as follows:
\begin{align*}
	\ p(J=j, M^i_j \neq \CA^i_j, M^i_\ell \neq \CA^i_\ell,r,b) &= \frac{q(J=j, M^i_j \neq \CA^i_j, M^i_\ell \neq \CA^i_\ell,\CE,r,b~|~\CE\cup\CF)}{q(\CE~|~\CE\cup\CF)}\\
	\ &\le \frac{20}{6}{q(J=j, M^i_j \neq \CA^i_j, M^i_\ell \neq \CA^i_\ell,\CE,r,b~|~\CE\cup\CF)} 
\end{align*}
where the last inequality follows since $q(\CE~|~\CE\cup\CF) = \frac{(4+\beta)+(2+\beta)}{(2+\beta)+(4+\beta)+3(2+\beta)+3\beta} = \frac{6+2\beta}{12+8\beta} \ge \frac{6}{20}$ as $\beta < 1$. 

Plugging the above in \eqref{eqn:midstep}, 
\begin{align}
	\  p(J=j, M^i_j \neq \CA^i_j,r,b) &\le 400{q(J=j, M^i_j \neq \CA^i_j, M^i_\ell \neq \CA^i_\ell,\CE,r,b~|~\CE\cup\CF)}. 
	\label{eqn:midstep1}
\end{align}

Note that the distribution $q(jm|rb,~\CE\cup\CF)$ is still product as $J$ and $M$ are independent given $rb$ and given the partition $b$ one can check if the matching $m$ is bad or consistent without knowing the cut. Next, we relate the probability of the right hand side to the probability under the partition $b'$ where $b'$ is the partition guaranteed from $(\ell,r,b)$ being in $\CG_1$. 
We will show that
\begin{align}
\ q(M^i_j \neq \CA^i_j, M^i_\ell \neq \CA^i_\ell,\CE~|rb,~\CE\cup\CF) \le 120{q(M^i_j \neq \CA^i_j, M^i_\ell \neq \CA^i_\ell,\CE~|~rb',\CE\cup\CF)}.
\label{eqn:bad}
\end{align}

Note that the two events above correspond to different matchings since the partitions are different (see Figure \ref{fig:mata}). The probability of any matching that agrees with the event on the right hand side above is zero under $p(m|rb)$, but that is not the case under $q(m|rb,\CE \cup \CF)$. In fact, by symmetry the probability of any such matching is the same under the distribution $q(m|rb,\CE\cup\CF)$ and $q(m|rb',\CE\cup\CF)$. This is because as $q(jm|rb,\CE\cup\CF)$ is product, the probability does not change under conditioning on the event that $J=\ell$ which is the same as conditioning that the cut is $b_0 \cup b^i_\ell$. Then, since the partitions were picked uniformly, both the splits $(b_0,b^i_\ell)$ and $(b'_0, b'^i_{\ell})$ are equally likely even when conditioned on $J=\ell, \CE\cup\CF$ since the matching is either consistent or bad with respect to both of them. 

\begin{figure*}[h]
    \centering
    \begin{subfigure}[h]{0.8\textwidth}
        \centering
        \includegraphics[height=1.4in]{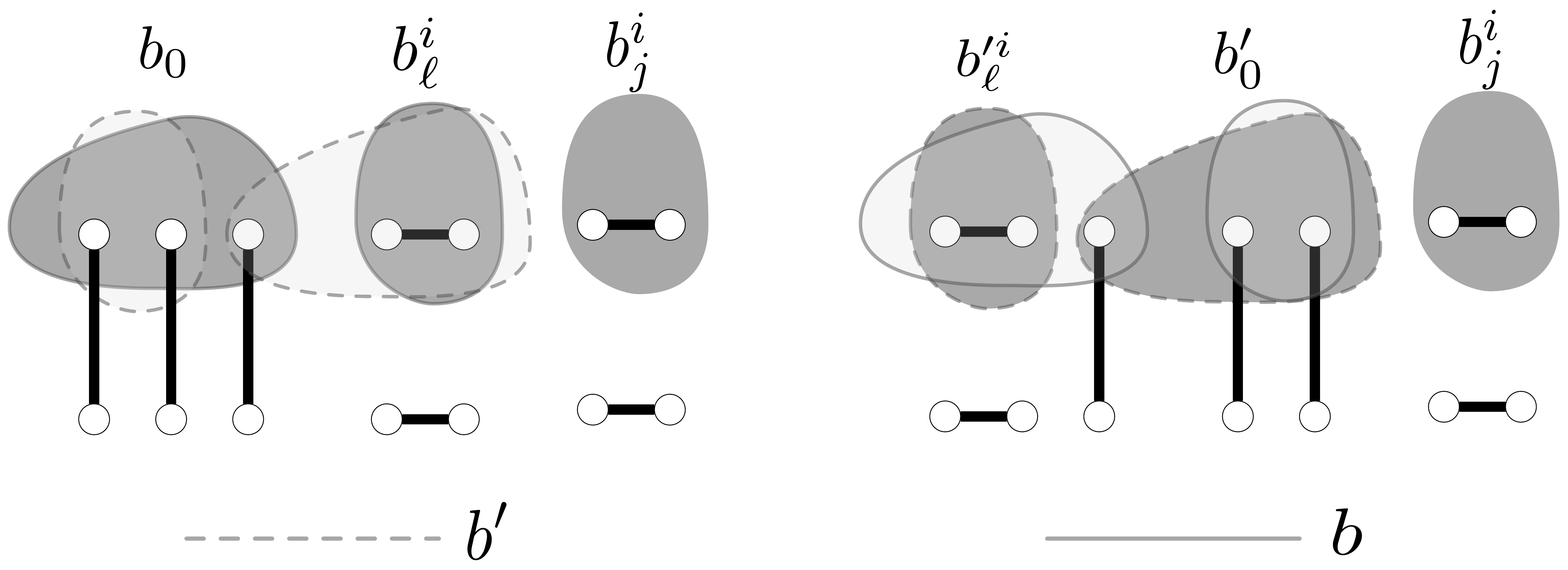}
        \caption{\footnotesize Partial Matchings corresponding to the event $M^i_j \neq \CA^i_j, M^i_\ell \neq \CA^i_\ell$ under partition $b$(left) and under partition $b'$ (right). Note that $b$ and $b'$ agree in all blocks except $b_0$ and $b^i_\ell$.}
		\label{fig:mata}
    \end{subfigure}%
	\vspace*{1cm}	
    \begin{subfigure}[h]{0.8\textwidth}
        \centering
        \includegraphics[height=1.4in]{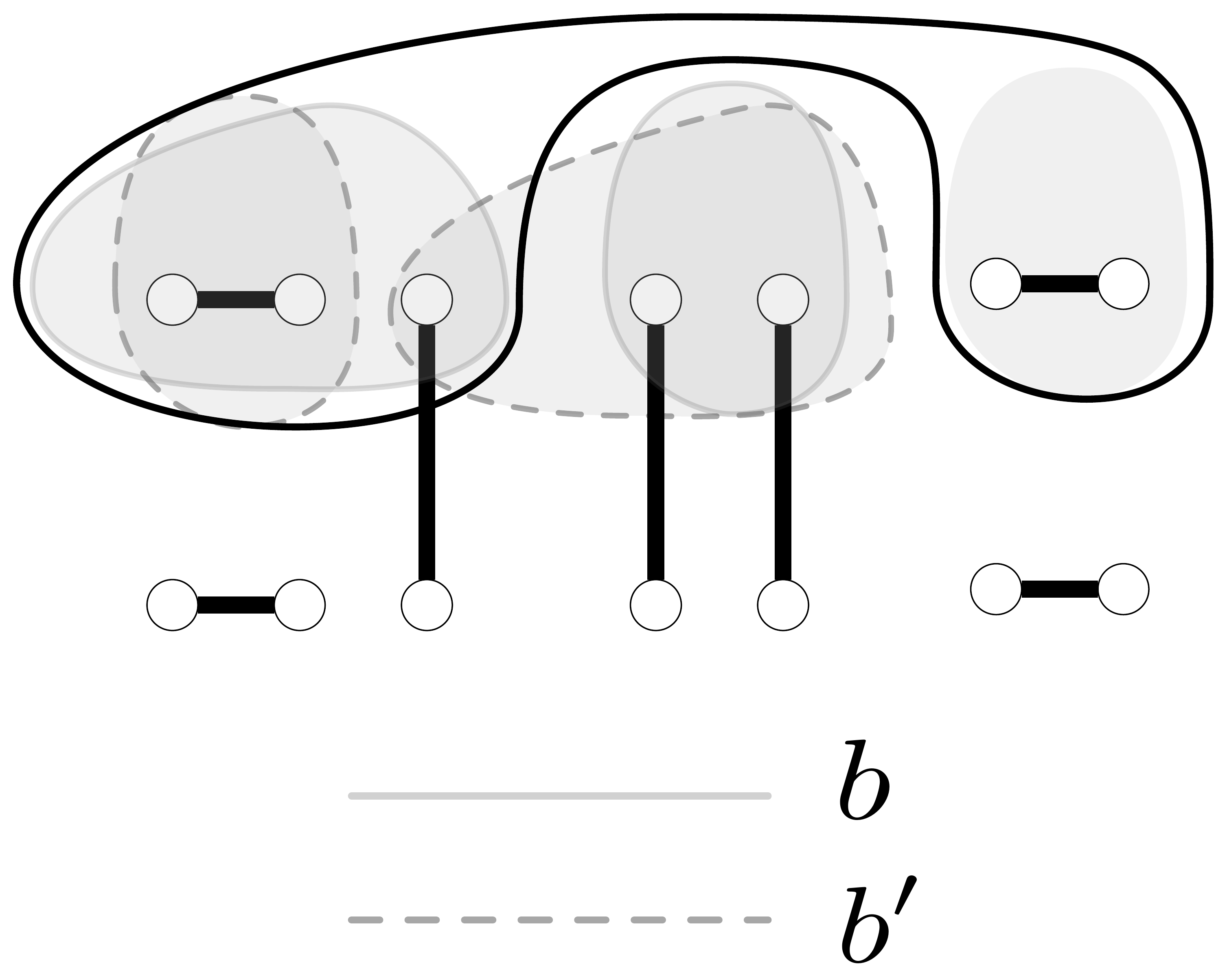}
		\caption{\footnotesize A matching agreeing with the event $M^i_j \neq \CA^i_j, M^i_\ell \neq \CA^i_\ell,\CE$ under partition $b'$ and the cut $U^i = b^i_0 \cup b^i_j$. Note that the matching is bad for the partition $b$.}
		\label{fig:matb}
    \end{subfigure}
	\label{fig:mat}
\end{figure*}

Note that any matching which agrees with the event on the right hand side above is bad for $b$ (see Figure \ref{fig:matb}). From this it follows that
\[ q(M^i_j \neq \CA^i_j, M^i_\ell \neq \CA^i_\ell,\CE~|rb,~\CE\cup\CF) \le {120}{q(M^i_j \neq \CA^i_j, M \text{ is bad }|~rb,\CE\cup\CF)}.\]

Plugging the above in \eqref{eqn:midstep1} and using that the distribution is product, we get \eqref{eqn:error}. To complete the proof, we next show that \eqref{eqn:bad} holds.

For this let us recall that since $(\ell,r,b)$ and $(\ell,r,b')$ are both in $\CG$, when $j \notin \CT_{rb}$, we have $p(M^i_j \neq \CA^i_j, M^i_\ell \neq \CA^i_\ell|rb)\ge\frac{1}{120}$ and a similar statement also holds with respect to $b'$.

Since $p(m|rb) = q(m|rb,\CE)$, we can write 
\[ \frac{p(M^i_j \neq \CA^i_j, M^i_\ell \neq \CA^i_\ell|rb)}{p(M^i_\ell = \CA^i_\ell|rb)} = \frac{q(M^i_j \neq \CA^i_j, M^i_\ell \neq \CA^i_\ell,\CE|~rb,\CE\cup\CF)}{q(M^i_\ell=\CA^i_\ell,\CE|~rb',\CE\cup\CF)}, \]
and a similar statement holds for $b'$ also.

Similar to the argument before, by symmetry $q(M^i_\ell = \CA^i_\ell,\CE|rb) =  q(M^i_\ell = \CA^i_\ell,\CE|rb')$ and also $p(M^i_\ell = \CA^i_\ell|rb) =  p(M^i_\ell = \CA^i_\ell|rb')$. It follows that
\begin{align*}
	\ \frac{q(M^i_j \neq \CA^i_j, M^i_\ell \neq \CA^i_\ell,\CE~|~rb,\CE\cup\CF)}{q(M^i_j\neq\CA^i_j, M^i_\ell\neq\CA^i_\ell,\CE~|~rb',\CE\cup\CF)} &= \frac{p(M^i_j \neq \CA^i_j, M^i_\ell \neq \CA^i_\ell|rb)}{p(M^i_j \neq \CA^i_j, M^i_\ell \neq \CA^i_\ell|rb')} \le 120,
\end{align*}
which proves \eqref{eqn:bad}. This completes the proof. 

\end{proof}

\begin{proof}[Proof of Claim \ref{claim:error}]
	We have $\BHc{p}(M^i|U^iB\CD) \ge t - 1$ since when $U^i=B^i_j$ then $M^i_k$ either equals $\CA^i_k$ or is not equal to $\CA^i_k$ with probability $\frac12$ independently for each $k \neq j$ conditioned on $\CD$. Moreover, conditioned on $\CD$, $U^i$ is uniform among blocks $B^i_j$ such that $M^i_j\neq \CA^i_j$. We may compute from \eqref{eqn:slack} that $p(m^i|b\CD) = \frac{\Delta(m^i)}{t2^{t-1}}$ where $\Delta(m^i)$ is the number of blocks $k$ of $m^i$ such that $m^i_k \neq \CA^i_k$. It follows that the probability of any $m^i$ with $\Delta(m^i) \le \frac{t}4$ is at most $\frac{1}{2^{t+1}}$ and using Proposition \ref{lemma:chernoff} (by considering the indicator vector for $M^i_k\neq \CA^i_k$ for $k \in [t]$) their total measure under the distribution $p(m^i|b\CD)$ can be bounded by $\frac{e^{-t/8}}2$. Hence, $\BHc{p}(U^i|M^iB\CD) \ge (1-\frac{e^{-t/8}}2) \log \left(\frac{t}{4}\right) \ge \log t - 3$.

As $p(umrb) = q(umrb|w\CE)$, if \eqref{eq:matinf} is not true, then  
	\[ \Infc[p]{R}{U^i}{M^iB\CD} = \BHc{p}(U^i|M^iB\CD) - \BHc{p}(U^i|RM^iB\CD) \le \beta^8\log t,\]
and a similar statement is obtained by writing $\Infc[p]{R}{M^i}{U^iB\CD}$ in terms of entropy. It follows that
\begin{align*}
	\ \BHc{p}(U^i|RBM^i\CD) &\ge (1-\beta^8)\log t-3, \text{~and~} \BHc{p}(M^i|RBU^i\CD) \ge t - \beta^8 \log t-1. 
\end{align*}

In terms of the random variable $J$ (recall that $J=j$ if $U^i=B^i_j$), we get 
\begin{align*}
	\ \BHc{p}(J|RBM^i\CD) &\ge (1-\beta^8)\log t-3, \text{~and~} \BHc{p}(M^i|JRB\CD) \ge t - \beta^8 \log t-1. 
\end{align*}

For rest of the proof, it will be helpful to keep in mind that under the distribution $p$, $\CD$ is equivalent to the event that $M^i_J \neq \CA^i_J$.

\begin{enumerate}[label=\emph{(\alph*)}]
	\item Follows from definition of $\CS_1$. 
	
	\item In the probability space $p$, define random variable $Y \in \bits^t$ as follows: for $k \in [t]$, $Y_k=\ind_{M^i_k = \CA^i_k}$. Then, it holds that $\BHc{p}(J|RB,Y_J=0) \ge (1-\beta^8)\log t-3$ and $\BHc{p}(Y|RB,Y_J=0) \ge t-\beta^8\log t-3$ and $J \arr RB \arr Y$. Applying Lemma \ref{lemma:main} gives us $p(\CS_2,\CD) \le p(\CS_2) \le 64\beta$. 

	\item Since removing conditioning only increases entropy, we have that $\BHc{p}(U^i|RB\CD) \ge (1-\beta^8)\log t-3$. Let $\CT = \{(r,b)~|~\BHc{p}(U^i|rb\CD) \ge (1-\beta^4)\log t\}$. As $\BHc{p}(U^i|rb\CD) \le \log t$, Lemma \ref{lemma:avg} then says that $p((r,b)\notin \CT|\CD) \le \beta^4 + \frac{3}{\beta^4\log t} \le 2\beta^4$ where the last inequality follows as $\beta^8\log t \ge 3$ when $\beta$ is a constant and $t \ge c=c(\beta)$ for a sufficiently large constant $c$.
		
		Lemma \ref{lemma:entsupp} implies that for any $(r,b)\in \CT$, $p(\CS_3|rb\CD) \le \beta^2$. By union bound, \[p(\CS_3,\CD) \le p(\CS_3|(r,b)\in \CT, \CD) + p((r,b) \notin \CT|\CD) \le \beta^2 + 2\beta^4 \le 3\beta^2.\]

	\item Set $a = \beta \frac{t^{1-\beta^2}}{2^{1/\beta^2+1}}  - 1$. Note that for any $k \in \CJ_{jrb}$, $\BHc{p}(M^i_k|rb,M^i_j\neq\CA^i_j) \le \sh\left(\frac14\right) \le \frac{9}{10}$. When $(j,r,b) \in \CS_4$, $|\CJ_{jrb}| \ge a$ and hence $\BHc{p}(M^i|rb,M^i_j\neq\CA^i_j) \le t - \frac{a}{10}$. Since, $p(jm|rb)$ is product, $\BHc{p}(M^i|rb,J=j,M^i_j\neq\CA^i_j) = \BHc{p}(M^i|rb,M^i_j\neq\CA^i_j)$ and we can say that 
		\begin{align*}
			\ \BHc{p}(M^i|JRB\CD) &= \BE_{p(jrb|\CD)} [\BHc{p}(M^i|rb,M^i_j\neq\CA^i_j)] \ge t - \beta^8\log t-1. 
		\end{align*}
		Using Lemma \ref{lemma:avg}, we get that $p(\CS_4,\CD) \le p(\CS_4|\CD) \le \frac{\beta^8\log t +1}{{a}/{10}} \le \beta$ as $t \ge c = c(\beta)$ for a large enough $c$.
\end{enumerate}

\end{proof}

\section{Proof of Main Technical Lemma (Lemma \ref{lemma:main})}
\label{sec:mainlemmaproof}

For the sake of contradiction, we assume that $p((x,r)\in \CB) \ge 64\gamma$. The proof will proceed by first fixing a value of $r$, such that the entropies $\BHc{p}(X|r, Y_X=0)$ and $\BHc{p}(Y|r, Y_X=0)$ will remain large but the distribution of $X$ and $Y$ conditioned on $r$ will be quite biased (conditioned on $r$, the probability of the event $Y_X=0$ will be significantly larger than $\frac12$). Then, we will argue that if this was the case, then in fact, the entropy $\BHc{p}(Y|r, Y_X=0)$ must be much smaller than our assumption.

Let $\CG$ denote the set of $r$ such that  
	\begin{enumerate}[label=\emph{(\alph*)}]
		\label{defn:goodr}
		\item $\BHc{p}(X|r, Y_X=0) \ge (1-\gamma^4)\log m$ and $\BHc{p}(Y|r, Y_X=0) \ge m - \gamma^4 \log m.$
		\item $p(x\in \CS_r|r) \ge 32\gamma$ where $\CS_r = \{x|(x,r) \in \CB\}.$
	\end{enumerate}

We will be able to argue that 
\begin{claim}
	\label{claim:goodr}
	$p(r \in \CG|Y_X=0) \ge 256\gamma^2$.
\end{claim}

In particular, this means that the set $\CG$ is not empty. For the rest of the proof, we will fix some $r \in \CG$ and work with the distribution $q(xy) = p(xy|r)$ which is product. Consider the random variable $\Bias = \ind_{Y_{X}=0} - \ind_{Y_{X} = 1}$. Note that when $x \in \CS_r$, 
\[ \BE_{q(xy|X=x)}[\Bias]  = q(Y_x=0) - q(Y_x=1) \ge \gamma. \]

We will prove that there exists a \emph{rectangle} with the following properties. 

\begin{claim}
	\label{claim:rec}
	There exists events $\CS \subseteq \supp(X)$, $\CT \subseteq \supp(Y)$ such that
	\begin{enumerate}[label={(\alph*)}]
		\item $\BE_{q(xy|(X,Y) \in \CS \times \CT)}[\Bias] \ge \gamma/2$.
		\item $q(x|X\in \CS) \le \frac{2^{2/{\gamma^2}}}{m^{1-\gamma^2}}$ for every $x$.
		\item $\BHc{q}(Y|Y \in \CT) \ge m - \gamma^2\log m - \frac{2}{\gamma^2}$.
	\end{enumerate}
\end{claim}

Lets finish the proof of Lemma \ref{lemma:main} first. For this, we define
\[\beta_{x} : = q(Y_{x}=0|Y \in \CT) - q(Y_{x}=1|Y \in \CT).\]

Since $q(xy|(X,Y) \in \CS \times \CT)$ is a product distribution, Claim \ref{claim:rec}(b) implies that 
\begin{align*}
	\ \BE_{q(xy|(X,Y) \in \CS \times \CT)}[\Bias] = \sum_{x} q(x|X\in \CS) \beta_{x} \le \frac{2^{2/{\gamma^2}}}{m^{1-\gamma^2}} \sum_{x \in \CS} \beta_{x}.
\end{align*}
So, using Claim \ref{claim:rec}(a), we can say
\[ \sum_{x \in \CS} \beta_{x} \ge \frac{\gamma}{2^{2/{\gamma^2}+1}} m^{1-\gamma^2}, \text{ or equivalently, } \BE_{x \in \CS} [\beta_{x}] \ge \frac{\gamma}{2^{2/{\gamma^2}+1}} \frac{m^{1-\gamma^2}}{|\CS|}, \]
From Lemma \ref{lemma:bias}, it then follows that
\[ \BHc{q}(Y|Y \in \CT) \le m - \left(\frac{\gamma}{2^{2/{\gamma^2}+2}} \frac{m^{1-\gamma^2}}{|\CS|}\right)^2 |\CS| = m - \frac{\gamma^2}{2^{4/{\gamma^2}+4}} \frac{m^{2-2\gamma^2}}{|\CS|}.\]
As Claim \ref{claim:rec}(b) also implies that $|\CS| \ge \frac{1}{2^{2/{\gamma^2}}} m^{1-\gamma^2}$, we have
\[ \BHc{q}(Y|Y \in \CT) \le m - \frac{\gamma^2}{2^{2/{\gamma^2}+4}} m^{1-\gamma^2} ,\]
which contradicts Claim \ref{claim:rec}(c) if 
\[ \frac{\gamma^2}{2^{2/{\gamma^2}+4}} m^{1-\gamma^2} > \gamma^2\log m + \frac{2}{\gamma^2}.\]

One can check that if $\frac{3}{\log m} \le \gamma^8 \le \frac{1}{2^{64}}$, then the left hand side above is always at least $\Omega(\sqrt{m})$ while the right hand side is $O(\log m)$. This proves Lemma \ref{lemma:main}.

Next we turn to proving Claims \ref{claim:goodr} and \ref{claim:rec}. We will need the following proposition.

\begin{prop} 
	\label{prop:cond}
	For any events $\CA$ and $\CB$, $p(\CA|\CB,Y_X=0) \ge p(\CA|\CB) p(Y_X=0|\CA\CB)$.
\end{prop}

\begin{proof}[Proof of Claim \ref{claim:goodr}]
	Let 
	\begin{align*}
		\ \CG_1 &= \{ r ~|~ \BHc{p}(X|r,Y_X=0) \ge (1-\gamma^4)\log m \text{ and } \BHc{p}(Y|r,Y_X=0) \ge m - \gamma^4 \log m\}, \text{ and}\\
		\ \CG_2  &= \{r~|~ p(x \in \CS_r|r) \ge 32\gamma\}.
	\end{align*}

	As $\BHc{p}(X|r,Y_X=0) \le \log m$ and $\BHc{p}(Y|r,Y_X=0) \le m$ for every $r$, Lemma \ref{lemma:avg} and union bound imply that $p(r \notin \CG_1|Y_X=0) \le 2\left(\gamma^4+\frac{3}{\gamma^4\log m}\right)\le 4\gamma^4$ as $\gamma^8 \ge \frac{3}{\log m}$.
	
	Lemma \ref{lemma:avg} implies that $p(r \in \CG_2) \ge 32\gamma$. Furthermore, for any $r \in \CG_2$, we have 
	\[ p(Y_X=0|r) \ge p(x \in \CS_r|r) p(Y_X=0|r,X\in \CS_r) \ge 16\gamma.\]

	Since $\gamma \le \frac1{2^8}$, using Proposition \ref{prop:cond}, 
	\begin{align*} 
		\ p(r\in \CG|Y_X=0) &\ge p(r \in \CG_2|Y_X=0) - p(r \notin \CG_1|Y_X=0) \\
		\				  &\ge p(r\in\CG_2) \cdot 16\gamma - p(r \notin \CG_1|Y_X=0) \ge 512\gamma^2 - 4 \gamma^4 \ge 256\gamma^2.
	\end{align*}
\end{proof}

\begin{proof}[Proof of Claim \ref{claim:rec}]
	By our choice of $r \in \CG$, we have $q(x\in \CS_r) \ge 32\gamma$ where $\CS_r = \{x|q(Y_x=0) \ge \frac{1+\gamma}2\}$. This also implies: 
	\begin{equation}
	 \ q(Y_X=0) \ge q(X \in \CS_r) q(Y_X=0|X\in \CS_r) \ge 16\gamma.
		\label{eqn:prob}
	\end{equation}
	
	We define
	\begin{align*}
		\ \CS_u &= \{x ~|~ q(x|Y_X=0) \le \frac{m^{\gamma^2} 2^{1/{\gamma^2}}}{m}\} \text{ and } \CS = \CS_r \cap \CS_u,\\
		\ \CT_b &= \{y ~|~ \BE_{q(xy|Y=y,X\in \CS)}[\Bias] \ge \frac{\gamma}2\}, \CT_u = \{y ~|~ q(y|Y_X=0) \le \frac{m^{\gamma^2} 2^{1/{\gamma^2}}}{2^m}\} \text{ and } \CT = \CT_b \cap \CT_u.
	\end{align*}

	By definition, we have $\BE_{q(xy|(X,Y) \in \CS \times \CT)}[\Bias] \ge \frac{\gamma}{2}$ which establishes (a).
	
	We shall prove the following propositions. 
	\begin{prop}
		\label{eq:measurex}
		For $x\in \CS_r$, we have $q(x) \le 2q(x|Y_X=0)$. Also, $q(x\in \CS|Y_X=0) \ge 8\gamma.$
	\end{prop}
		
	\begin{prop}
		\label{eq:measurey}
		For $y \in \CT_b$, we have $q(y) \le \frac{1}{4\gamma}q(y|Y_X=0)$. Also, $q(y\in\CT|Y_X=0) \ge \gamma^2.$
	\end{prop}
	
	Proposition \ref{eq:measurex} and \eqref{eqn:prob} then implies, $q(x\in\CS) \ge q(Y_X=0)q(x\in\CS|Y_X=0)\ge 128\gamma^2$, and since $\gamma \le \frac{1}{2^8}$, we get that
	\[q(x|X\in \CS) = \frac{q(x)}{q(x\in \CS)} \le \frac{2q(x|Y_X=0)}{q(x\in \CS)} \le \frac{m^{\gamma^2} 2^{1/{\gamma^2}}}{64\gamma^2 m} \le \frac{2^{2/{\gamma^2}}}{m^{1-\gamma^2}}.\]
	holds for every $x\in \CS$ which proves (b). 

	With a similar argument, using Proposition \ref{eq:measurey} we get that for $y\in\CT$, the following holds
	\[q(y|Y\in \CT) \le \frac{m^{\gamma^2}2^{1/{\gamma^2}}}{64\gamma^4 2^m} \le \frac{m^{\gamma^2}2^{2/{\gamma^2}}}{2^m}.\]
	which implies that $\BHc{q}(Y|Y\in\CT) \ge m - \gamma^2\log m - \frac{2}{\gamma^2}$ which gives (c).
\end{proof}

	Now we turn to proving Propositions \ref{eq:measurex} and \ref{eq:measurey}. 
	
	\begin{proof}[Proof of Proposition \ref{eq:measurex}]
		Recall that $q(Y_x=0) \ge \frac{1+\gamma}2$ for $x \in \CS_r$ and $q(x\in \CS_r) \ge 32\gamma$. Since $q(xy) = p(xy|r)$, using Proposition \ref{prop:cond}, it follows that $q(x) \le 2q(x|Y_X=0)$ for $x\in \CS_r$ and hence $q(x\in \CS_r|Y_X=0) \ge 16\gamma$. 

	The entropy bound $\BHc{q}(X|Y_X=0) \ge (1-\gamma^4)\log m$ implies that $q(x \in \CS_u|Y_X=0) \ge 1 - \gamma^2$ using Lemma \ref{lemma:entsupp}. 
	
	Since $\gamma \le \frac{1}{2^8}$, we get 
	\[q(x \in \CS|Y_X=0) \ge q(x\in \CS_r|Y_X=0) - q(x\in\comp\CS_u|Y_X=0) \ge 16\gamma - \gamma^2 \ge 8\gamma.\]
	\end{proof}

	\begin{proof}[Proof of Proposition \ref{eq:measurey}]
		Note that $\BE_{q(xy|Y=y,X\in \CS)}[\Bias] \ge \gamma/2$ is equivalent to saying that $q(Y_X=0|y,X\in\CS) \ge \frac12 + \frac{\gamma}{4}$. As $q(xy) = p(xy|r)$ using Proposition \ref{prop:cond}, it follows that $q(y|X\in\CS) \le 2q(y|X\in\CS,Y_X=0)$ for $y\in\CT_b$. Since $q(xy)$ is product, Bayes rule and Proposition \ref{eq:measurex} imply that for $y\in\CT_b$, the following holds
	\begin{align*}
		\ q(y)  = q(y|X\in\CS) & \le 2q(y|X\in\CS,Y_X=0) \\
		\     				 &=\frac{2q(y|Y_X=0)q(x\in\CS|y,Y_X=0)}{q(x\in\CS|Y_X=0)} \le \frac{1}{4\gamma}q(y|Y_X=0).
	\end{align*}
	
	By definition of $\CS$, $p(Y_X=0|X\in \CS) \ge \frac12 + \frac{\gamma}{2}$, or equivalently, $\BE_{q(xy|X\in \CS)}[\Bias] \ge \gamma$. Lemma \ref{lemma:avg} then implies that $q(y\in\CT_b) \ge \gamma/2$ and hence $q(y\in\CT_b|Y_X=0) \ge 4\gamma q(y\in\CT_b) \ge 2\gamma^2$. 
	
	 Also, since $\BHc{q}(Y|Y_X=0) \ge m-\gamma^4 \log m$, using Lemma \ref{lemma:entsupp} we have $q(y \in \CT_u|Y_X=0) \ge 1 - \gamma^2$.  Hence, we can conclude 
	\[q(y \in \CT|Y_X=0) \ge q(y\in\CT_b|Y_X=0) - q(x\in\comp\CT_u|Y_X=0) \ge \gamma^2.\]  
	\end{proof}

\section{Acknowledgments}

Thanks to Paul Beame, Siva Ramamoorthy, Anup Rao and Thomas Rothvo{\ss} for valuable discussions and feedback on the writing. Further thanks to Anup and Siva for help with the figures, and to anonymous referees for helpful comments.

\bibliographystyle{alpha}
{{\bibliography{matching}}}

\newcommand{\etalchar}[1]{$^{#1}$}
\begin{thebibliography}{MNSW98}

\bibitem[AIP06]{AIP06}
Alexandr Andoni, Piotr Indyk, and Mihai Patrascu.
\newblock On the optimality of the dimensionality reduction method.
\newblock In {\em Proceedings of the 47th Annual IEEE Symposium on Foundations
  of Computer Science}, FOCS '06, pages 449--458, Washington, DC, USA, 2006.
  IEEE Computer Society.

\bibitem[AT13]{AT13}
David Avis and Hans~Raj Tiwary.
\newblock On the {E}xtension {C}omplexity of {C}ombinatorial {P}olytopes.
\newblock In {\em Automata, Languages, and Programming - 40th International
  Colloquium, {ICALP} 2013, Riga, Latvia, July 8-12, 2013, Proceedings, Part
  {I}}, pages 57--68, 2013.

\bibitem[BFPS15]{BFPS15}
G{\'{a}}bor Braun, Samuel Fiorini, Sebastian Pokutta, and David Steurer.
\newblock Approximation {L}imits of {L}inear {P}rograms ({B}eyond
  {H}ierarchies).
\newblock {\em Math. Oper. Res.}, 40(3):756--772, 2015.

\bibitem[Bie08]{B08}
Daniel Bienstock.
\newblock Approximate formulations for 0-1 knapsack sets.
\newblock {\em Oper. Res. Lett.}, 36(3):317--320, 2008.

\bibitem[BM13]{BM13}
Mark Braverman and Ankur Moitra.
\newblock An information complexity approach to extended formulations.
\newblock In {\em Symposium on Theory of Computing Conference, STOC'13, Palo
  Alto, CA, USA, June 1-4, 2013}, pages 161--170, 2013.

\bibitem[BP15]{BP15}
G{\'{a}}bor Braun and Sebastian Pokutta.
\newblock The {M}atching {P}roblem {H}as {N}o {F}ully {P}olynomial {S}ize
  {L}inear {P}rogramming {R}elaxation {S}chemes.
\newblock {\em {IEEE} Trans. Information Theory}, 61(10):5754--5764, 2015.

\bibitem[BP16]{BP16}
G{\'{a}}bor Braun and Sebastian Pokutta.
\newblock Common {I}nformation and {U}nique {D}isjointness.
\newblock {\em Algorithmica}, 76(3):597--629, 2016.

\bibitem[BR11]{BR11}
Mark Braverman and Anup Rao.
\newblock Information equals amortized communication.
\newblock In {\em FOCS}, pages 748--757, 2011.

\bibitem[CCZ10]{CCGZ10}
Michele Conforti, G{\'e}rard Cornu{\'e}jols, and Giacomo Zambelli.
\newblock Extended formulations in combinatorial optimization.
\newblock {\em 4OR}, 8(1):1--48, Mar 2010.

\bibitem[CLRS13]{CLRS13}
Siu~On Chan, James~R. Lee, Prasad Raghavendra, and David Steurer.
\newblock Approximate {C}onstraint {S}atisfaction {R}equires {L}arge {LP}
  {R}elaxations.
\newblock In {\em 54th Annual {IEEE} Symposium on Foundations of Computer
  Science, {FOCS} 2013, 26-29 October, 2013, Berkeley, CA, {USA}}, pages
  350--359, 2013.

\bibitem[CT06]{CT06}
Thomas~M. Cover and Joy~A. Thomas.
\newblock {\em Elements of Information Theory (Wiley Series in
  Telecommunications and Signal Processing)}.
\newblock Wiley-Interscience, 2006.

\bibitem[Edm65]{E65}
Jack Edmonds.
\newblock Maximum {M}atching and a {P}olyhedron with $0,1$ {V}ertices.
\newblock {\em J. of Res. the Nat. Bureau of Standards}, 69~B:125--130, 1965.

\bibitem[FMP{\etalchar{+}}15]{FMPTW15}
Samuel Fiorini, Serge Massar, Sebastian Pokutta, Hans~Raj Tiwary, and Ronald
  de~Wolf.
\newblock Exponential {L}ower bounds for {P}olytopes in {C}ombinatorial
  {O}ptimization.
\newblock {\em J. {ACM}}, 62(2):17, 2015.

\bibitem[GJW16]{GJW16}
Mika G{\"{o}}{\"{o}}s, Rahul Jain, and Thomas Watson.
\newblock Extension {C}omplexity of {I}ndependent {S}et {P}olytopes.
\newblock {\em CoRR}, abs/1604.07062, 2016.

\bibitem[H{\aa}s96]{H96}
Johan H{\aa}stad.
\newblock Clique is {H}ard to {A}pproximate {W}ithin $n^{1-\varepsilon}$.
\newblock In {\em 37th Annual Symposium on Foundations of Computer Science,
  {FOCS} '96, Burlington, Vermont, USA, 14-16 October, 1996}, pages 627--636,
  1996.

\bibitem[KMR17]{KMR17}
Pravesh~K. Kothari, Raghu Meka, and Prasad Raghavendra.
\newblock Approximating rectangles by juntas and weakly-exponential lower
  bounds for {LP} relaxations of csps.
\newblock In {\em Proceedings of the 49th Annual {ACM} {SIGACT} Symposium on
  Theory of Computing, {STOC} 2017, Montreal, QC, Canada, June 19-23, 2017},
  pages 590--603, 2017.

\bibitem[LRS15]{LRS15}
James~R. Lee, Prasad Raghavendra, and David Steurer.
\newblock Lower {B}ounds on the {S}ize of {S}emidefinite {P}rogramming
  {R}elaxations.
\newblock In {\em Proceedings of the Forty-Seventh Annual {ACM} on Symposium on
  Theory of Computing, {STOC} 2015, Portland, OR, USA, June 14-17, 2015}, pages
  567--576, 2015.

\bibitem[MNSW98]{MNSW98}
Peter~Bro Miltersen, Noam Nisan, Shmuel Safra, and Avi Wigderson.
\newblock On data structures and asymmetric communication complexity.
\newblock {\em Journal of Computer and System Sciences}, 57(1):37 -- 49, 1998.

\bibitem[NR15]{RR15}
Sivaramakrishnan {Natarajan Ramamoorthy} and Anup Rao.
\newblock How to {C}ompress {A}symmetric {C}ommunication.
\newblock In {\em 30th Conference on Computational Complexity, {CCC} 2015, June
  17-19, 2015, Portland, Oregon, {USA}}, pages 102--123, 2015.

\bibitem[Pat11]{P11}
Mihai Patrascu.
\newblock Unifying the landscape of cell-probe lower bounds.
\newblock {\em {SIAM} J. Comput.}, 40(3):827--847, 2011.

\bibitem[PV13]{PV13}
Sebastian Pokutta and Mathieu~Van Vyve.
\newblock A note on the extension complexity of the knapsack polytope.
\newblock {\em Oper. Res. Lett.}, 41(4):347--350, 2013.

\bibitem[Rot14]{R14}
Thomas Rothvo{\ss}.
\newblock The matching polytope has exponential extension complexity.
\newblock In {\em Symposium on Theory of Computing, {STOC} 2014, New York, NY,
  USA, May 31 - June 03, 2014}, pages 263--272, 2014.

\bibitem[Wil84]{W84}
Richard~M. Wilson.
\newblock The exact bound in the {E}rd{\H o}s-{K}o-{R}ado theorem.
\newblock {\em Combinatorica}, 4(2):247--257, 1984.

\bibitem[Yan91]{Y91}
Mihalis Yannakakis.
\newblock Expressing {C}ombinatorial {O}ptimization {P}roblems by {L}inear
  {P}rograms.
\newblock {\em J. Comput. Syst. Sci.}, 43(3):441--466, 1991.

\end{thebibliography}
\appendix

\section{Approximating the Matching Polytope}
\label{sec:upperboundmat}
We show that the polytope $Q_\eps(n)$ is a $(1+\eps)$ approximation for the matching polytope.

\begin{claim} 
	$P_M(n) \subseteq Q_\eps(n) \subseteq (1+\eps)P_M(n)$.
\end{claim}
\begin{proof} 
	First inclusion is trivial since any vector $x \in P_M(n)$ is also in $Q_{\eps}(n)$. To see the second inclusion, let $x \in Q_{\eps}(n)$ be an arbitrary vector. We only need to show that for any odd set $U \subseteq [n]$, $|U|> \frac{1+\eps}{\eps}$, $x(E(U)) \le (1+\eps)\frac{|U|-1}{2}$ since the other inequalities are already satisfied. Note that $x(E(U)) \le \frac{1}{2}\sum_{v \in U} x(\delta(v)) \le |U|/2 \le (1+\eps)\frac{|U|-1}{2}$ where the last inequality holds when $|U|\le \frac{1+\eps}{\eps}$.
\end{proof}

\section{Proofs of Preliminary Propositions and Lemmas}
\label{sec:prelimproofs}

\begin{proof}[Proof of Proposition \ref{prop:binent}]
	Define the function $\sg(x) = 1 - 2\log e\left(x - \frac12\right)^2 - \sh(x)$ for $x\in[0,1]$. Then, the first and second derivatives of $\sg(x)$ for $x \in [0,1]$ are 
	\[ \sg'(x) = -4 \log e\left(x-\frac12\right) - \log(1-x) - \log x \:\text{ and }\: \sg''(x) = -4\log e + \frac{\log e}{x(1-x)}.\]
As $x(1-x)\le \frac14$ for $x \in [0,1]$, $\sg''(x) \ge 0$ for $x\in [0,1]$. It follows that the function $\sg$ is convex on the interval $[0,1]$ and hence has a unique minima. Furthermore, the derivative $\sg'(\frac12) = 0$, so the minimum value of $\sg(x)$ is attained at $x=\frac12$. Hence, $\sg(x) \ge \sg(\frac12) = 0$.
\end{proof}

\begin{proof}[Proof of Proposition \ref{prop:nrank}]
	If $\nrank(A)=r$, then $A = \sum_{i=1}^r u_i^Tv_i$ where $u_i \in \BR_{\ge 0}^\CX$  and $v_i \in \BR_{\ge 0}^\CY$ are non-negative (column) vectors. Then, $p(xy)$ can be expressed as a convex combination of $r$ product distributions by setting
	\[ p(xy|r) = \frac{u_r(x)v_r(y)}{\sum_{x',y'} u_r(x')v_r(y')} \text{ and } p(r) = \frac{\sum_{x',y'} u_r(x')v_r(y')}{\sum_{x',y'}M_{x'y'}}.\]
\end{proof}

\begin{proof}[Proof of Lemma \ref{lemma:directsum}]
	Using the chain rule,
	\begin{align*}
		\ \sum_{i=1}^n \Infc{X_i}{R}{X_{<i}Y_{\ge i}} \le \sum_{i=1}^n \Infc{X_iY_{<i}}{R}{X_{<i}Y_{\ge i}} &= \sum_{i=1}^n \Infc{X_i}{Y_{<i}}{X_{<i}Y_{\ge i}} + \sum_{i=1}^n \Infc{X_i}{R}{X_{<i}Y}\\
		\								&= \sum_{i=1}^n \Infc{X_i}{R}{X_{<i}Y} = \Infc{X}{R}{Y}.
	\end{align*}
	where the second last equality follows since $\Infc{X_i}{Y_{<i}}{X_{<i}Y_{\ge i}} = 0$. The second bound follows similarly.
\end{proof}

\begin{proof}[Proof of Lemma \ref{lemma:intfam}]
	In case $|\cup_{\CF \in \FF} \CF|\le 4$, the size of the family $\FF$ is bounded by $\binom{4}{3}=4$. We will prove that if $|\cup_{\CF \in \FF} \CF|=5$, then size of $\FF$ can at most be $3$. Take sets $\CF_1, \CF_2, \CF_3$ such that $|\CF_1 \cup \CF_2|=4$ while $|\CF_1 \cup \CF_2 \cup \CF_3|=5$. Then, $\CF_1 \cap \CF_2 \subseteq \CF_3$. But now there can be no more sets in the family because such a set must intersect $\CF_1 \cap \CF_2$ in just one element but in this case it will intersect with one of $\CF_1, \CF_2$ or $\CF_3$ in exactly one element.     
\end{proof}

\end{document}